%% file: main_arxiv_v1.tex
\documentclass{article}
\input{setting_arxiv_v1}

\title{Weight distribution bounds to
relate minimum distance, list decoding,
and symmetric channel performance}

\usepackage{authblk}
\author{Donald Kougang-Yombi}
\author{Jan H{\k{a}}z{\l}a}

\affil{AIMS Rwanda\\ \texttt{\{donald.yombi, jan.hazla\}@aims.ac.rw}}

\date{}
\begin{document}

\maketitle
\begin{abstract}
    We study relationships between worst-case and
    random-noise properties of error correcting codes.
    More concretely, we consider connections between minimum distance,
    list decoding radius, and block error probability on
    noisy channels.

    A recent result of Pernice, Sprumont, and Wootters established
    the tight connection between list decoding radius and symmetric
    channel performance for linear codes. We extend this result to general codes. The proof proceeds by directly bounding the weight  distribution rather than by sharp threshold techniques.

    We then turn to the relation between minimum distance
    and symmetric channel performance.
    The results we just described imply that a $q$-ary code of relative distance
    $\delta$ has vanishing error probability on the symmetric
    channel up to the Johnson radius $J_q(\delta)$.
    We improve upon this bound in the case of linear codes,
    for a range $\delta$, for $q\ge 4$.
    In our proof we consider the \emph{erasure} properties
    of codes, and bound their weight distribution through
    inequalities introduced by Samorodnitsky. The latter
    part of the proof gives a more general technique that
    bounds the symmetric channel performance 
    of a linear code with constant relative distance and good
    erasure channel performance.
 \end{abstract}
\section{Introduction}
\subsection{List decoding and symmetric channels}
\label{sec:intro-list-decoding}

An error correcting code with finite input alphabet $\Sigma$ and block length $n$ is a subset $\cC\subseteq\Sigma^n$. When
$|\Sigma|=q$ we say that the code is $q$-ary (and binary for $q=2$).
A code can be thought of as a packing of $|\cC|$ possible messages
that allows to recover the original message even if the
sent codeword $x\in \cC$ is corrupted by noise.
Typically, one desires codes which have a good tradeoff
between the rate $R(\cC)\coloneqq \frac{\log_q |\cC|}{n}$
(which quantifies how many different messages can
be encoded into $\cC$) and
the level of noise they can recover from.

The concept of maximum recoverable level of noise can
be made concrete in several meaningfully different ways.
In the \emph{worst-case model}, we consider the \emph{minimum
distance} $d(\cC)\coloneqq \min_{x,y\in\cC, x\ne y} \Delta(x,y)$,
where $\Delta(x,y)$ denotes the \emph{Hamming distance}, that is
the number of coordinates $1\le i\le n$ such that $x_i\ne y_i$.
It is well known that for a code with minimum distance $d$ 
there exists a decoder that recovers the original codeword
for every pattern of $\lceil d/2\rceil-1$ errors and
every pattern of $d-1$ erasures. For a family
of codes with increasing block lengths $\{\cC_n\}_n$,
we can express how well it corrects errors
in the worst-case model by the notion of
\emph{relative minimum distance} 
$\delta\coloneqq \liminf_{n\to\infty}\delta(\cC_n)$,
where $\delta(\cC_n)\coloneqq d(\cC_n)/n$.
Accordingly, such a family corrects up to around
$\delta/2$ fraction of worst case errors
and $\delta$ fraction of worst case erasures.

In the \emph{probabilistic (Shannon)} error model,
we consider a (memoryless) communication channel
which independently distorts every coordinate $1\le i\le n$ of the original codeword according to a fixed probability 
distribution. Accordingly, a channel $W$ with output
alphabet $\mathcal{Y}$ is characterized
by a family of $q$ probability distributions on $\mathcal{Y}$,
one each for every input symbol $x\in\Sigma$.
The most natural example 
is the \emph{$q$-ary symmetric
channel} $\qSC_{p}$ with crossover probability
$0\le p\le 1$, which satisfies $W(x|x)=1-p$
and $W(x'|x)=\frac{p}{q-1}$ for $x\ne x'$.
Here we will focus on
the range $0\le p\le 1-1/q$, as at $p=1-1/q$
the channel is fully noisy and input symbols are statistically indistinguishable.
Another important example is the \emph{$q$-ary erasure
channel} $\qEC_\lambda$ with erasure probability $0\le\lambda\le 1$. This channel has output alphabet $\mathcal{Y}=\Sigma\cup\{?\}$ and satisfies
$W(x|x)=1-\lambda$ and $W(?|x)=\lambda$ for every $x\in\Sigma$.
In the Shannon model, it is easy to see that one cannot
recover the original message with certainty in any nontrivial scenario. Instead, a code
family $\cC_n$ 
is considered to successfully correct errors on a channel $W$ if
there exists a decoder which recovers the original message
except with probability that goes to $0$ as $n$ grows
(for every codeword, with respect to the channel noise;
or in expectation over a randomly chosen codeword).
In other words, we want the family $\cC_n$ to
have \emph{vanishing block error probability}.

The notion of $\emph{list decoding}$
is often considered a ``bridge'' between
the worst-case and Shannon models. For a more detailed
introduction to list decoding, see~\cite[Chapter 7]{guruswami2012essential}. A code is said to be
$(p, L)$-list decodable if for every $y\in\Sigma^n$
it holds $|\{x\in\cC:\Delta(y,x)\le pn\}|\le L$. In particular,
the decoder that, given $y$, outputs a list of $L$ closest
codewords to $y$ according to the Hamming distance,
will include the original codeword in the list for any
pattern of at most $pn$ errors. One can consider a 
family of codes $\cC_n$ to be ``good'' at list decoding with $p$ fraction of errors if it is $(p,L)$-list decodable for some modest $L=L(n)$, for example $L=O(1)$ or $L=\poly(n)$.

Several questions can be asked 
about general connections between 
the performance of codes for various
noise models. Perhaps the most natural is: 
What is the largest crossover probability $p$
such that every code family $\cC_n$ of relative distance
$\delta$ has vanishing block error probability on the $\qSC_p$?
We discuss this question
in more detail in
\Cref{sec:intro-distance-vs-qsc}.

As for the connection between the
minimum distance and
list decoding, a
classical result known as
the Johnson bound~\cite{Joh62} states that a code
of relative distance $\delta$ 
is $(p,q\delta n^2)$-list decodable for
\begin{align*}
p<J_q(\delta)\coloneqq \left(1-\frac{1}{q}\right)
\left(1-\sqrt{1-\frac{q}{q-1}\delta}\right)\;,
\end{align*}
We will call the quantity $J_q(\delta)$ the \emph{($q$-ary) Johnson radius}.
The Johnson bound is tight in the sense that there exist linear codes
with relative distance $\delta$ and Hamming balls of radius
larger than $J_q(\delta)$ with superpolynomially many codewords
(see \cite[Exercise 7.9]{guruswami2012essential}).

Recently, Pernice, Sprumont, and Wootters~\cite{pernice2025list} studied the
connection between list decoding and
the Shannon model. A \emph{linear code} 
$q$-ary code
for a prime power $q$ is a linear subspace of the
vector space $\mathbb{F}_q^n$.
If a code
is $(p,L)$-list decodable,
then one can attempt unique decoding on $\qSC_{p'}$ for $p'<p$ 
by outputting a random element of the list of
size $L$ which
(with high probability)
contains the original codeword. Accordingly,
this randomized decoder 
has average block error probability
at most $1-1/L+\exp(-\Omega(n))$.
For $q=2$ and a linear code,
the classical sharp threshold
result for block error probability by Tillich and Zémor
\cite{tillich2000discrete} can be used 
to establish vanishing block error probability
on the $\qSC_{p'}$,
if the minimum distance of
the code is sufficiently large as a function of $L$.

\cite{pernice2025list}
generalized this sharp threshold result to linear codes over larger alphabets.
As a result, they established that if a code family
is $(p,L)$-list decodable and satisfies\footnote{The assumption
on the minimum distance is mild,
but necessary, as discussed in~\cite{pernice2025list}.} 
$d(\cC_n)\ge\frac{q^3}{(1-p)^2}\ln^4(nL)$,
then it has vanishing error probability
on $\qSC_{p'}$
for every $p'<p$. In other words,
list decodability up to radius $p$
implies vanishing error probability on
the symmetric channel
up to crossover probability $p$.
The relation between the list decoding radius
and the crossover probability is tight.
This is certified by random codes (and random linear codes),
which are known both to achieve
list decoding capacity~\cite{Eli91, GHK10},
and, by the Shannon's coding theorem,
achieve capacity on the symmetric channel.

The first theorem we present is a generalization of the
result of Pernice, Sprumont, and Wootters
to general codes. As a result, we obtain 
the tight relation between list decoding radius
and crossover probability for \emph{all}
error correcting codes with any fixed alphabet size.

\begin{theorem}
\label{thm:list-vs-qsc-main}
Let $0<p\le 1-1/q$ and 
$\{\cC_n\}_n\subseteq\Sigma^n$ be a family
of $q$-ary codes 
which is $(p,L)$-list decodable
for some $L=L(n)$ satisfying
$d(\cC_n)=\omega(\log(nL))$.
 Then,
$\lim_{n\to\infty} P_B(\cC_n,\qSC_{p'})=0$
for every $p'<p$.
\end{theorem}

Above, $P_B(\cC,W)$ denotes the MAP block error
probability for code $\cC$ on channel $W$,
worst case over the codewords. We give the precise
definition in \Cref{sec:preliminaries}.

The proof of \Cref{thm:list-vs-qsc-main}
does not use a sharp threshold technique.
Instead, we use a simple double counting 
argument to show that $(p,L)$-list decoding implies
a bound on the weight distribution of a code for weights
$w\ge pn$. Given a code $\cC$, its weight distribution
for $0\le w\le n$ and $x\in\Sigma^n$ is given as
\begin{align*}
A_w(x)\coloneqq
\left|\left\{
y\in\cC:\Delta(x,y)=w
\right\}\right|\;.
\end{align*}

Let $S(x,t)$ denote the Hamming sphere of radius $t$ centered at $x\in\Sigma^n$. Then,
let $\nu_q(n,t_1,t_2,w)\coloneqq|S(0^n,t_1)\cap S(1^w0^{n-w},t_2)|$, 
in other words $\nu_q(n,t_1,t_2,w)$ denotes the
size of the intersection of two $q$-ary Hamming spheres
with radii $t_1,t_2$, such that their centers
are at the Hamming distance $w$.

\begin{theorem}
\label{thm:list-double-counting}
Let $\cC\subseteq\Sigma^n$ be a $q$-ary code which is $(p,L)$-list decodable and let $0\le w,t_1\le n$, $0\le t_2\le pn$, and $x\in\Sigma^n$. Then,
\begin{align}
    A_w(x)\cdot\nu_q(n,t_1,t_2,w)&\le 
    \binom{n}{t_1}(q-1)^{t_1}L\;.
    \label{eq:double-counting-bound}
\end{align}
\end{theorem}

\begin{proof}
Consider the set
\begin{align*}
S&\coloneqq\left\{
(y,z): z\in\cC, \Delta(y,x)=t_1, \Delta(y,z)=t_2
\right\}\;.
\end{align*}
On the one hand, for every $y\in\Sigma^n$ such
that $\Delta(y,x)=t_1$, by the list decoding property 
there exist at most $L$ codewords $z$ such that 
$\Delta(y,z)=t_2$.
Hence,
\begin{align}
\label{eq:double_count-upper}
|S|\le \binom{n}{t_1}(q-1)^{t_1} L\;.
\end{align}
On the other hand, for every $z\in\cC$ such that
$\Delta(x,z)=w$, we have $(y,z)\in S$ if and only if $\Delta(y,x)=t_1$ and $\Delta(y,z)=t_2$, in other words if
$y\in S(x,t_1)\cap S(z,t_2)$. Therefore,
\begin{align}
\label{eq:double_count-lower}
|S|\ge A_{w}(x)\nu_q(n,t_1,t_2,w)\;.
\end{align}
The result follows by combining \Cref{eq:double_count-lower,eq:double_count-upper}.
\end{proof}

If we invoke \Cref{thm:list-double-counting}
for $t\coloneqq t_1=t_2=\lfloor pn\rfloor$, the expression
$\nu_q(n,t,w)\coloneqq
\nu_q(n,t,t,w)$ becomes the size of
the intersection of two Hamming spheres of the same
radius $t$ and centers at distance $w$.
A well known upper bound on the symmetric channel
block error probability by Poltyrev~\cite{poltyrev2002bounds, pathegama2023smoothing} can be written using
sizes of intersections of two Hamming \emph{balls} of the
same radius. As these sizes have equal exponents up to lower order terms,
\Cref{thm:list-double-counting} can be combined
with the Poltyrev bound in a rather clean manner
in order to establish
\Cref{thm:list-vs-qsc-main}.
In particular, codes that achieve list decoding
capacity also achieve capacity on the symmetric channel
purely due to their weight distribution properties.\footnote{
As a contrasting example, in the case of Reed--Muller codes,
despite extensive research into their weight distribution
(see \cite{ASSY23}), the several known 
proofs
that they achieve 
capacity~\cite{kudekar2016reed,reeves2023reed,abbe2023proof,abbe2023reed,PSZ25}
do not rely 
directly on the weight distribution.
}

We note that even though the choice $t_1=t_2$ 
is convenient for our proof
(and sufficient for the tight dependence in
\Cref{thm:list-vs-qsc-main}), it does not in general
achieve the tightest possible weight distribution 
bound in~\eqref{eq:double-counting-bound}.

\subsection{Symmetric channel decoding
from minimum distance}
\label{sec:intro-distance-vs-qsc}

As mentioned, perhaps
the most natural theoretical question 
relating the worst-case and random-noise properties of error correcting
codes is: For a given relative distance $\delta$, 
what is the largest crossover probability $p$
such that every 
family of $q$-ary codes $\{\cC_n\}_n$ with relative distance
$\delta$ has vanishing block error probability on the symmetric channel $\qSC_p$?
Accordingly, given a $q$-ary alphabet $\Sigma$ and relative minimum distance 
$0\le\delta\le 1-1/q$, let
\begin{align*}
    \psym(q,\delta)&\coloneqq\sup\left\{p\ge 0: \forall\ \{\cC_n\}_n,
    \cC_n\subseteq\Sigma^n,
    \text{ if }\delta(\cC_n)\ge \delta\text{, then } P_B\left(\cC_n,\qSC_p\right)\to 0\right\}\;.
\end{align*}
For $q$ a prime power, the same question can be studied restricted to linear codes, so let us define
\begin{align*}
\pLsym(q,\delta)&\coloneqq\sup\left\{p\ge 0: \forall\ \text{linear}\ \{\cC_n\}_n, \cC_n\subseteq\mathbb{F}_q^n,
    \text{ if } \delta(\cC_n)\ge \delta \text{, then } P_B\left(\cC_n,\qSC_p\right)\to 0\right\}\;.
\end{align*}

As we are not aware of an extensive discussion of this problem in the existing literature, we proceed to a short survey of what we know about it. Recall the $q$-ary entropy function
$H_q(x):=x\log_q(q-1)-x\log_q(x)-(1-x)\log_q(1-x)$.
A random code (as well as a random linear code) 
of rate $1-H_q(\delta)-\eps$ satisfies $\delta(\cC_n)\ge \delta$.
At the same time, by the Shannon's channel coding theorem, it cannot
have vanishing error probability on $\qSC_{p}$ if
$H_q(p)>H_q(\delta)+\eps$. Taking $\eps\to 0$, we have an upper bound
\begin{align}
    \label{eq:psym-upper-bound}
    \psym(q,\delta)\le\pLsym(q, \delta)\le \delta\;.
\end{align}

Since for larger values of $q$ we know
code families that beat the Gilbert--Varshamov bound on the minimum distance,
the upper bound $\pLsym(q,\delta)\le\delta$ is not always
tight. Consider the construction of
algebraic geometric codes due to Tsfasman, Vl{\u{a}}duţ, and Zink  \cite{tsfasman1982modular,ihara1981some,couvreur2021algebraic}. When $q$ is a square, there exist families of such linear $q$-ary codes with rate and relative distance satisfying $R\ge 1-\delta-\frac{1}{\sqrt{q}-1}$. Combining this with the Shannon's channel coding theorem, it follows that such codes cannot achieve vanishing error probability on the $\qSC_p$ whenever $H_q(p)>\delta+\frac{1}{\sqrt{q}-1}$, therefore, for $\delta\le 1-\frac{1}{\sqrt{q}-1}$,
\begin{align*}
    \pLsym(q, \delta)\le H_q^{-1}\left(\delta+\frac{1}{\sqrt{q}-1}\right)\;.
\end{align*}
This improves upon \Cref{eq:psym-upper-bound} in certain regimes, in particular in some range of $\delta$ for relatively large $q$  ($q\ge 49$).

Let us turn to a more detailed discusion of lower bounds for $\psym(q,\delta)$.
Since a code with relative distance $\delta$ can recover from
$\delta n/2-1$ errors in the worst-case model, it is easy to establish
a trivial bound
\begin{align*}
    \psym(q,\delta)\ge \delta/2\;.
\end{align*}
This bound can be improved by combining
the Johnson bound with the result
from \cite{pernice2025list}
for linear codes and
now with \Cref{thm:list-vs-qsc-main}
for general codes, giving
\begin{align*}
\psym(q,\delta)\ge J_q(\delta)\;,
\end{align*}
which is a strict improvement
as it is elementary to check
that $\delta/2< J_q(\delta)$
for every $q\ge 2$ and $0<\delta< 1-1/q$.

For large alphabet sizes, there is also a remarkable result
by Rudra and Uurtamo~\cite{rudra2010two}.
They show that for fixed $\delta>0$ and $\epsilon>0$, 
a code family with relative distance $\delta$ and alphabet
size $q\ge 2^{\Omega(\frac{1}{\varepsilon})}$
satisfies $\lim_{n\to\infty}P_B(\cC_n,\qSC_{\delta-\eps})=0$.
Together with the upper bound~\eqref{eq:psym-upper-bound},
that implies
\begin{align*}
\lim_{q\to\infty}\psym(q,\delta)=\lim_{q\to\infty}\pLsym(q,\delta)=\delta\;.
\end{align*}
However, this result does not seem to give nontrivial bounds for
smaller values of $q$. In particular, 
an effective version of the bound that we extracted from their proof
(see \Cref{app:effective-ru})
does not seem to improve upon the Johnson
bound for any value of $\delta$ for $q\le 19$.\footnote{
Furthermore, it seems that for each fixed $q$, their bound
is below the Johnson bound for sufficiently large $\delta$.
This is exactly the region where our \Cref{thm:unconditional_bound} shows improvement
over the Johnson bound for $q\ge 4$.
}

It seems natural to ask if the barrier
of the Johnson radius can be overcome
for small values of $q$.
As our second contribution, we demonstrate that
$\pLsym(q,\delta)>J_q(\delta)$ for $q\ge 4$
and for large enough $\delta$.

Let us proceed to a short sketch of our proof.
Since both the Johnson bound and \Cref{thm:list-vs-qsc-main}
are tight, an improvement on the Johnson radius
cannot be achieved by exploiting solely the list
decoding properties of a code with a given relative distance. 

We proceed in a seemingly roundabout way,
by considering erasures instead of errors.
Consider the analogous notion of list decoding from erasures:
A code $\cC\subseteq\Sigma^n$ is 
\emph{$(\lambda,L)$-erasure list decodable} if,
for every $S\subseteq[n]$ with $|S|> (1-\lambda)n$
and every $y\in\Sigma^S$ it holds
$\left|\left\{c\in \cC:c_S=y\right\}\right|\le L$.
There is a version of the Johnson bound for erasures
which says that a $q$-ary code family with relative distance $\delta'>\delta$
is 
$\left(\frac{q}{q-1}\delta,O(1)\right)$-erasure list decodable.
At the same time, the result in~\cite{pernice2025list}
also has its erasure version: A code family with linear minimum
distance which
is $(\lambda,\poly(n))$-erasure list decodable has vanishing error probability
for $\qEC_{\lambda'}$ with $\lambda'<\lambda$.
Altogether, a code family of relative distance $\delta$
has vanishing error probability for $\qEC_\lambda$
for $\lambda<\frac{q}{q-1}\delta$, which improves upon the trivial
range of $\lambda<\delta$.

Inequalities by Samorodnitsky~\cite{samorodnitsky2019upper,samorodnitsky2020improved},
which were recently generalized to larger alphabet sizes~\cite{abawonse2025generalized},
give weight distribution bounds for linear codes
that have small error probability on the
erasure channel. Those bounds (and the minimum distance assumption)
can then be plugged into the Poltyrev bound 
for the symmetric channel error probability.
It turns out that this improves on the Johnson
bound for larger values of $\delta$ for $q\ge 4$.

Before stating this result, we will discuss how we combine Samorodnitsky's inequalities
with the minimum distance assumption. This is
a more general technique that can be applied for
any linear code with linear minimum distance and
good erasure channel performance.

\subsection{Symmetric channel decoding from
minimum distance and erasure channel}
\label{sec:intro-sam}

Consider a binary linear code
$\mathcal{C}\subseteq\mathbb{F}_2^n$ with
relative distance $\delta(\mathcal{C})\ge 0.1$.
Applying the Johnson bound as discussed in \Cref{sec:intro-list-decoding},
such a code has small error probability on
$\BSC_{0.052}$, which is slightly better than
the trivial bound of $\BSC_{0.05}$.
On the other hand, consider another linear code that has
small block error probability on the binary erasure
channel
$\BEC_{0.533}$. By the application of
Samorodnitsky inequalities from~\cite{hkazla2021codes}, such a code
will also have small error probability
on the $\BSC_{0.052}$.
But what if a linear code satisfies both properties: the minimum distance bound
$\delta(\mathcal{C})\ge 0.1$, and
vanishing block error probability on $\BEC_{0.533}$?
We will show that then one can obtain a stronger
conclusion: $\mathcal{C}$ has small block error
probability on $\BSC_{0.077}$.

We study this problem systematically, obtaining bounds
for MAP block error probability for linear codes
and various channels,
in terms of the minimum distance and the performance on the
erasure channel. We do this for the symmetric channel,
both in the binary case and for larger finite fields,
and for the binary Gaussian channel $\BAWGN_{\sigma^2}$.
In general, we demonstrate that assuming
both conditions
(minimum distance and erasure channel performance) results in
better bounds than when only one of the conditions is present.

Our results apply the bound on the weight
distribution in terms of the erasure channel performance given in~\cite{hkazla2021codes} for $q=2$ and~\cite{abawonse2025generalized} for larger alphabets.
Both of these works then apply the Bhattacharayya (union) bound
to bound the error probability for various communication channels. Let us state the asymptotic version of those
results. 
We use $P_b(\cC,W)$ and $P_B(\cC,W)$ 
to denote, respectively the (MAP) bit and block error
probabilities of a code $\cC$ on channel $W$
and $Z(W)$ to denote the Bhattacharayya coefficient of $W$, see~\Cref{sec:preliminaries} for the definitions.
\begin{theorem}[\cite{hkazla2021codes,abawonse2025generalized}]
\label{thm:original}
Let $0<\lambda< 1$, and $W:\mathbb{F}_q\to\mathcal{Y}$ be
a single-letter channel with finite field input alphabet satisfying
$Z(W)<\frac{q^{\lambda}-1}{q-1}$. Let
$\{\cC_n\}_n$ be a family of linear codes
$\cC_n\subseteq\mathbb{F}_q^n$ satisfying
$d(\cC_n)\ge\omega(\log n)$ and 
$\lim_{n\to\infty} P_B(\cC_n,\qEC_{\lambda})=0$.
Then, 
$\lim_{n\to\infty}P_B(\cC_n,W)=0$.
\end{theorem}

We focus on the symmetric channel $\qSC_p$,
and, for $q=2$,
we also consider the binary Gaussian channel $\BAWGN_{\sigma^2}$ for $\sigma\ge 0$, where conditioned on an input $x\in\mathbb{F}_2$, the output distribution
is $\mathcal{N}((-1)^x,\sigma^2)$. For those channels
the Bhattacharyya coefficients are given by
$Z(\qSC_p)=\frac{q-2}{q-1}p+2\sqrt{\frac{p(1-p)}{q-1}}$
and $Z(\BAWGN_{\sigma^2})\allowbreak =\exp(-1/(2\sigma^2))$.
The condition $Z(W)<\frac{q^{\lambda}-1}{q-1}$ 
in \Cref{thm:original} holds for a corresponding
range of parameters, which in the binary case
reduces to
$|p-1/2|> \sqrt{2^{\lambda-1}(1-2^{\lambda-1})}$
and $\sigma^2<
-\frac{1}{2\ln(2^\lambda-1)}$.

We prove a generalization of
\Cref{thm:original} for symmetric channels.
We need to make two definitions, one for
the binary case and one for $q\ge 3$.

\begin{definition}
\label{def:delta-star-binary}
Here and later, 
let $h(x)\coloneqq -x\log_2 x-(1-x)\log_2(1-x)$ 
be the binary entropy function.
For $0\le\gamma\le 2p\le 1$, let
\begin{align}
M(\gamma,p)
&\coloneqq
\gamma+(1-\gamma)h\left(\frac{p-\frac{\gamma}{2}}{1-\gamma}\right)\;,
\label{eq:11}\\
F(\gamma, p)
& \coloneqq 
h(p)-M(\gamma,p)\;,
\nonumber
\end{align}
where $M$ is continuously
extended to
$M(1,1/2)\coloneqq 1$.
Then, for $0<\lambda\le 1$ and
$0\le \delta\le 1$, let
\begin{align*}
 p_{*}(\lambda,\delta)\coloneqq
 \inf\left\{\delta/2\le p\le 1/2: F(\delta,p)\le -\delta\log_2(2^{\lambda}-1)\right\}\;.
\end{align*}
\end{definition}

\begin{definition}
\label{def:q_ary_functions}
Let $q\ge 3$.
Recall the $q$-ary entropy function
$H_q(x)=-(1-x)\log_q(1-x)-x\log_q\frac{x}{q-1}$.
Furthermore,
let\footnote{
In other words, $\Htilde_q(x)$ is the (base-$q$)
entropy of the distribution on $q$ elements with
probabilities $(1-x)/2,(1-x)/2,x/(q-2),\ldots,
x/(q-2)$.
} $\Htilde_q(x)\coloneqq
-(1-x)\log_q\frac{1-x}{2}-x\log_q\frac{x}{q-2}$.
For $0\le\gamma\le \min(2p,1)$ and
$0\le p\le 1-1/q$,
let
\begin{align*}
        D(\gamma,p)&\coloneqq\Big\{(\zeta,\beta)\in[0,1]^2: \gamma\frac{\zeta}{2}+(1-\gamma)\beta\le p-\frac{\gamma}{2}\Big \}\;,\\
M^{(q)}(\gamma,p)&\coloneqq \underset{\substack{(\zeta,\beta)\in
D(\gamma,p)}}{\sup}\gamma\tilde{H}_q(\zeta)+(1-\gamma)H_q(\beta)\;,\\
F^{(q)}(\gamma, p)
&\coloneqq 
H_q(p)-M^{(q)}(\gamma,p)\;.
\end{align*}
Then, for $0<\lambda\le 1$ and
$0\le \delta\le 1$, let
\begin{align*}
 p^{(q)}_{*}(\lambda,\delta)\coloneqq
 \inf\left\{\delta/2\le p\le 1-1/q: F^{(q)}(\delta,p)\le 
 \delta\log_q\left(\frac{q-1}{q^\lambda-1}\right)\right\}\;.
\end{align*}
\end{definition}

Since $F(\delta,1/2)=0$, the value of
$p_{*}(\lambda,\delta)$ is well defined.
Similarly, for $q\ge 3$ we check that
$(1-2/q,1-1/q)\in D(\delta,1-1/q)$, therefore
$M^{(q)}(\delta,1-1/q)=1$,
hence $F^{(q)}(\delta,1-1/q)=0$, and
the value of $p_{*}^{(q)}(\lambda,\delta)$ is well defined.

We are ready to state our result
that generalizes \Cref{thm:original}.

\begin{theorem}
\label{Thm:1}
   Let $0<\delta,\lambda<1$, and let $\{\mathcal{C}_n\}_{n\ge1}$ be a family of $q$-ary linear codes with $\delta(\cC_n)\ge\delta$, such that $\lim_{n\to\infty}P_b(\cC_n,\qEC_{\lambda})=0$.
   Then,  
   \begin{align*}
       \underset{n\to\infty}{\lim}P_B(\cC_n,\qSC_p)=0,\ \text{for every}\ p<p_{*}^{(q)}(\lambda,\delta)\;,
   \end{align*}
   where in the binary case we write
   $p_*^{(2)}\coloneqq p_*$.
\end{theorem}

Note that in our result we assume only vanishing
\emph{bit} error probability on the erasure channel. This is due to the
assumption of linear minimum distance, and in contrast to 
\cite{hkazla2021codes}, where in general stronger assumptions 
on the bit error probability are required (see their Theorem 3.1).
Furthermore, due to the 
sharp threshold results for the block error probability discussed above, the rate of
convergence of the block error probability in \Cref{Thm:1}
is exponential in $n$.

To prove \Cref{Thm:1}, we use the same weight distribution
bounds as introduced in \cite{hkazla2021codes,abawonse2025generalized}.
However, instead of using the union 
(Bhattacharyya) bound
as therein, we use the assumption $\delta(\cC_n)\ge\delta$ and a
version of the tighter 
Poltyrev bound~\cite{poltyrev2002bounds,pathegama2023smoothing}
for the symmetric channel error probability in terms of the weight distribution.
The Poltyrev bound is expressed in terms of the cardinalities of
intersections 
$\mu_q(n,t,w)\coloneqq|B(0^n,t)\cap B(1^w0^{n-w},t)|$, where $B(x,t)$ is the Hamming ball
centered at $x$ and with radius $t$. 

As will be seen later, $M^{(q)}(\gamma,p)$ is related
to $\mu_q(n,t,w)$ in the sense that 
$M^{(q)}(\gamma,p)=\lim_{n\to\infty}\frac{1}{n}\log_q(\mu_q(n,\allowbreak pn,\lfloor \gamma n\rfloor)$.
The most technical part of our proofs is
the analysis of various properties of
functions $F^{(q)}$ and
$M^{(q)}$. This is not too difficult for $q=2$,
but more involved for larger alphabets, since
in that case $M^{(q)}$ does not seem to have a closed form. Among others,
we show that $M^{(q)}(\gamma,p)$
is continuous, that it is concave in $\gamma$, 
and that $F^{(q)}(\gamma,p)$ is decreasing in $p$.
We also derive an ``almost closed form'' for 
$M^{(q)}$ which allows to compute its values easily.

Finally, we instantiate our technique on the binary Gaussian channel
$\BAWGN_{\sigma^2}$.
In this case, we replace the union Bhattacharyya bound with the ``sphere bound''~\cite{HP94,sason2006performance}. 
In all cases, adding the minimum distance assumption gives strictly better results for any fixed relative distance 
$\delta>0$, and the resulting bounds
converge to those from
\cite{hkazla2021codes,abawonse2025generalized} as
$\delta\to 0$.

\subsection{Improving on the Johnson bound}

Let us return to the problem
of bounding symmetric channel
performance for linear codes of
relative distance $\delta$.
As we described in~\Cref{sec:intro-distance-vs-qsc},
a $q$-ary code family of relative distance
$\delta$ has vanishing block error probability
on $\qEC_\lambda$ for $\lambda<\frac{q}{q-1}\delta$.
Plugging this into~\Cref{Thm:1}, we show:

\begin{theorem}
\label{thm:unconditional_bound}
For $q\ge 2$ a prime power and $0<\delta\le 1-1/q$, 
it holds
$\pLsym(q,\delta)\ge p^{(q)}_*\left(\frac{q\delta}{q-1}, \delta\right)$.
\end{theorem}

As we mentioned, it can be observed numerically 
that this improves somewhat
on the Johnson bound for $q\ge 4$ and larger values
of $\delta$. We include an illustration for $q=9$ and $q=17$ in the whole range $[0,1-1/q]$ (see \Cref{fig:pstar_vs_jq}), and for $q\in\{3,4,9,17\}$ for larger values of $\delta$ (see \Cref{fig:pstar_vs_jq_large_delta} in the appendix). 

\begin{figure}[!ht]
    \centering
    \includegraphics[width=0.8\linewidth]{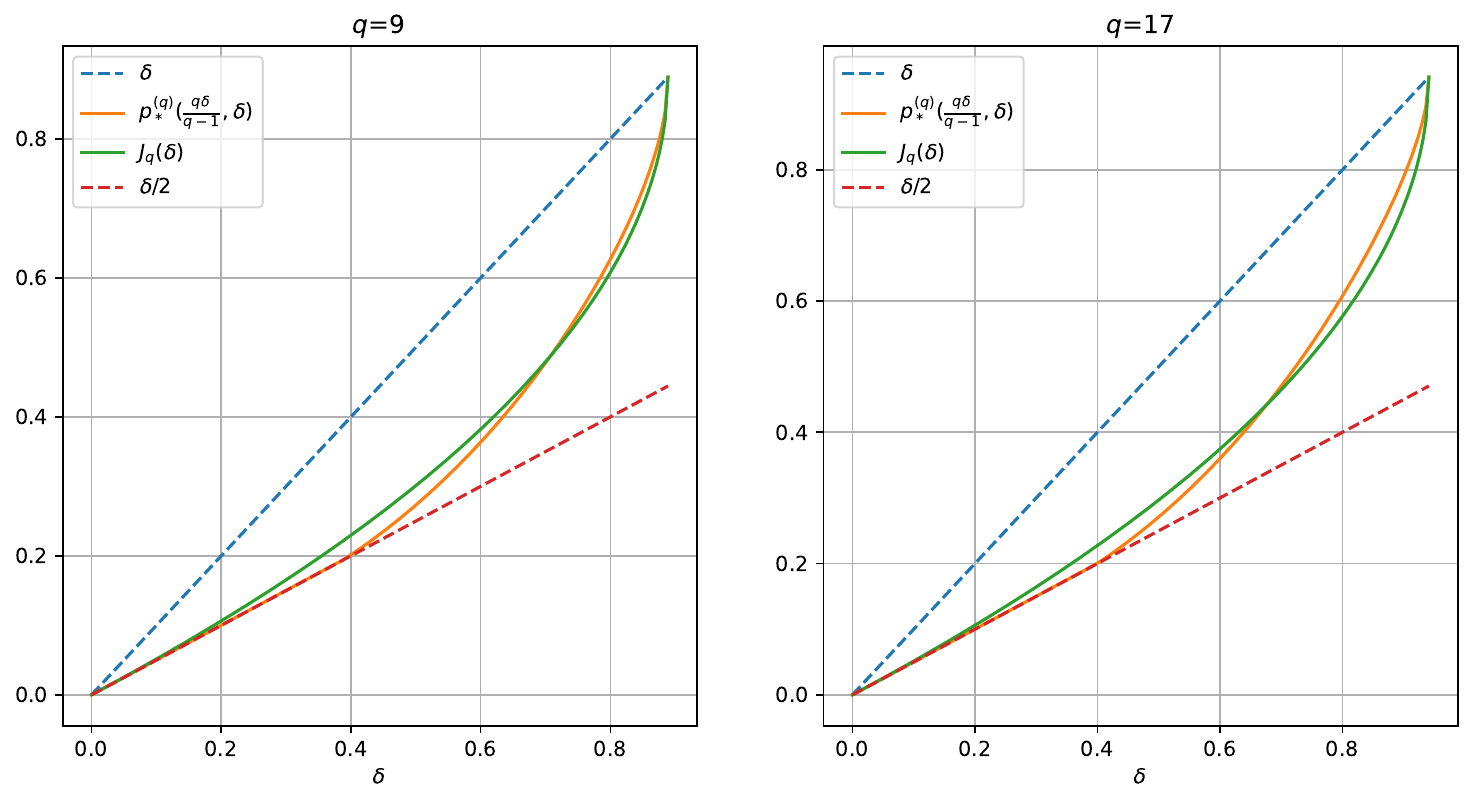}
    \caption{Our lower bound $p_*^{(q)}(\frac{q\delta}{q-1},\delta)$ from \Cref{thm:unconditional_bound} 
    plotted against the Johnson bound for
    $q=9$ and $q=17$. Our bound is larger for larger values of $\delta$. 
    Also the upper bound
    of $\delta$ and the lower bound of $\delta/2$ are plotted
    for illustration.}
    \label{fig:pstar_vs_jq}
\end{figure}

\paragraph{Outline of the paper.}
\Cref{sec:preliminaries} recaps our notation and preliminaries.
\Cref{sec:proof-of-list-qsc} contains the proof
of \Cref{thm:list-vs-qsc-main}.
In \Cref{sec:samorodnitsky} we prove and discuss
\Cref{Thm:1}, and \Cref{sec:dual} discusses
a similar bound in terms of the
erasure channel performance of the \emph{dual code}.
\Cref{sec:unconditional_bound} contains the proof
of \Cref{thm:unconditional_bound},
and in \Cref{sec:m-f-proofs} we prove various
properties of $M^{(q)}$ and $F^{(q)}$ deferred from the
preceding sections. \Cref{sec:bawgn} discusses
the analogue of \Cref{Thm:1} for the BAWGN channel.

\section{Preliminaries}
\label{sec:preliminaries}
Throughout the paper, $\log(\cdot)$ without specification will denote the logarithm in base 2. 
We let $H_q(x)\coloneqq -(1-x)\log_q(1-x)-x\log_q\frac{x}{q-1}$ be the $q$-ary entropy function. For $q=2$, we denote the binary entropy function $H_2(x)$ by $h(x)$. Furthermore, for $q\ge 3$
we let 
$\Htilde_q(x)\coloneqq -(1-x)\log_q\frac{1-x}{2}
-x\log_q\frac{x}{q-2}$.
Recall that the $q$-ary
entropy function $H_q(\beta)$ is
a concave function increasing
for $0\le \beta\le1-1/q$ and
decreasing for $1-1/q\le\beta\le 1$.
A similar claim about the function
$\Htilde_q(\beta)$ is easily checked:
\begin{claim}
     \label{clm:3}
     $\tilde{H}_q(\beta)$ is a concave function,
     increasing for $0\le\beta\le1-2/q$, and 
     decreasing for $1-2/q\le\beta\le 1$.
\end{claim}

We denote the $q$-ary finite alphabet
by $\Sigma$, so that $|\Sigma|=q$.
From time to time for convenience we make
the identification $\Sigma=\mathbb{Z}_q$.
For $q$ prime power and linear codes we take
$\Sigma=\mathbb{F}_q$.

We denote the $q$-ary Hamming distance by $\Delta(x,y)$.
For $\Sigma=\mathbb{Z}_q$ we write $\wt(x)\coloneqq \Delta(x,0^n)$.
The $q$-ary Hamming ball (respectively, sphere) of center $x\in\Sigma^n$ and radius $t\in\mathbb{R}_+$ (respectively, $t\in\mathbb{N}$) is denoted as $B(x,t)\coloneqq\{y\in\Sigma^n: \Delta(x,y)\le t\}$ (respectively, $S(x,t)\coloneqq\{y\in\bF_q^n: \Delta(x,y)= t\}$).
We also let
$\mu_q(n,t,w):=|B(x,t)\cap B(y,t)|$ and $\nu_q(n,t_1,t_2,w):=|S(x,t_1)\cap S(y,t_2)|$ for some choice
of $x,y\in\Sigma^n$ with $\Delta(x,y)=w$.
(It is easy to see that the values
are independent of the choice of $x,y$.)
Then, we let $\nu_q(n,t,w)\coloneqq
\nu_q(n,t,t,w)$.

A $q$-ary code of block length $n$ is a 
nonempty subset $\cC\subseteq\Sigma^n$
and a linear code is a linear subspace of
$\mathbb{F}_q^n$. We denote the rate
$R(\cC)\coloneqq \frac{\log_q|\cC|}{n}$,
the minimum distance $d(\cC)\coloneqq\min_{x,y\in \cC: x\ne y}\Delta(x,y)$ and the relative distance
$\delta(\cC)\coloneqq \frac{d(\cC)}{n}$.
For $0\le w\le n$ and $x\in\Sigma^n$,
let $A_w(x)\coloneqq
\left|\left\{
y\in\cC:\Delta(x,y)=w
\right\}\right|$.
Then, we let 
$A_w\coloneqq \frac{1}{|\cC|}\sum_{x\in\cC}A_w(x)$.
The vector $(A_0,\ldots,A_n)$ is known as the
\emph{weight distribution} of a code.
Recall that for a linear code it holds
$A_w=A_w(x)$ for every $x\in\cC$.

\paragraph{List Decoding.}
For $x\in\Sigma^n$ and a subset 
$T\subseteq[n]\coloneqq\{1,\ldots,n\}$, we denote $x_T\coloneqq(x_i)_{i\in T}$. 
Let $\cC\subseteq\Sigma^n$ be a $q$-ary code, 
$0\le\rho\le 1$ and $L\ge 1$.
\begin{definition}
    The code $\cC$ is $(\rho,L)$-list decodable if for every $y\in\Sigma^n$ it holds
    $|\{x\in \cC:\Delta(x,y)\le\rho n\}|\le L$.
\end{definition}
\begin{definition}
    The code $\cC$ is $(\rho,L)$-erasure list decodable if,
for every $T\subseteq[n]$ with $|T|> (1-\rho)n$
and every $y\in\Sigma^T$, it holds
$|\{x\in \cC: x_T=y\}|\le L$.
\end{definition}
\paragraph{MAP decoding.}
Let $W$ be a channel with input alphabet $\Sigma$
and
output alphabet $\mathcal{Y}$. Assume we transmit uniformly at random a codeword $X\in\cC\subseteq\Sigma^n$ and denote by $Y\in\mathcal{Y}^n$ the output of transmitting $X$ through the $n$ 
independent uses of $W$. The maximum a posteriori (MAP) block decoding  is defined as 
\begin{align*}
    \hat{x}^{\MAP}(y):=\underset{x\in\cC}{\argmax} \ P_{X|Y}(x|y)\;,
\end{align*}
breaking the ties arbitrarily. (None of our results
depend on the choice of the tiebreaking method.)
In particular, for the $\qSC$ (respectively $\BAWGN$) one can verify that $\hat{x}^{\MAP}(y)$ is a closest codeword $c\in\cC$ to $y$ with respect to the Hamming distance (respectively the Euclidean distance). Similarly we can define the symbol-MAP (bit-MAP in $\bF_2$) decoding for $1\le i\le n$ as
$\hat{x_i}^{\MAP}(y):=\underset{x_i\in\Sigma}{\argmax}\ P_{X_i|Y}(x_i|y)$.
\begin{definition}[Error Probability]
For $x\in\cC$,
the block-MAP and symbol-MAP error probabilities are respectively defined as
\begin{align*}
    P_{B}(\cC,W,x)&\coloneqq
    \Pr\left[ \hat{x}^{\MAP}(Y)\neq X\;\vert\;X=x\right]\;,\
    \\P_{b}(\cC,W,x)
    &\coloneqq\frac{1}{n}\sum_{i=1}^n P_{b,i}(\cC,W,x)
    \coloneqq\frac{1}{n}\sum_{i=1}^n
    \Pr\left[ \hat{x_i}^{\MAP}(Y)\neq X_i
    \;\vert\;X=x\right]\;.
\end{align*}
We then let
$P_B(\cC,W)\coloneqq\max_{x\in\cC} P_B(\cC,W,x)$
and $P_b(\cC,W)\coloneqq\max_{x\in\cC}
P_b(\cC,W,x)$.
\end{definition}

\begin{definition}[Bhattacharyya coefficient]
\label{def:bhat-coeff}
Let $W$ be a channel with input
alphabet $\Sigma$.
Let $X$ be distributed uniformly in $\Sigma$
and $Y$ be the channel output
determined
by $P_{Y|X}=W$. Then, the Bhattacharayya coefficient
$Z(W)$ is given as
\begin{align*}
Z(W)&\coloneqq
\max_{x,x'\in\Sigma: x\ne x'}\E\left[
\sqrt{\frac{\Pr[X=x'\;\vert\;Y]}
{\Pr[X=x\;\vert\;Y]}}
\;\Big\vert\;X=x
\right]\;.
\end{align*}
\end{definition}

\section{Proof of \Cref{thm:list-vs-qsc-main}}
\label{sec:proof-of-list-qsc}

\subsection{The $M^{(q)}$ and $F^{(q)}$ functions}

In this section we will state some properties
of the $M^{(q)}$ and $F^{(q)}$ functions defined in
\Cref{sec:intro-sam}. Our main concern is their
relation to the sizes of the intersections of Hamming balls
and spheres.
We defer the proofs of the lemmas stated in
this subsection to \Cref{sec:m-f-proofs}.

Recall $\mu_q(n,t,w)=|B(0^n,t)\cap B(1^{w}0^{n-w},t)|$, the size of the intersection of two 
$q$-ary Hamming balls of radius $t$ whose centers differ in $w$ coordinates.
For some discussion on this 
quantity, see \cite[Section 3]{dong2023number}.
Here we establish the connection between
$M^{(q)}(w/n,t/n)$ and
$\mu_q(n,t,w)$ for $w,t=\Theta(n)$.

\begin{lemma}
    \label{lem:qary_balls_intersection}
    Let $q\ge 2$.
    Let $t=p n$ (not necessarily an integer)
    for $0\le p\le 1-1/q$
    and $w=\gamma n$ be an integer
    with $0\le w\le \min(2t,n)$. Then,
\begin{align*}
\mu_q(n,t,w)\le n^2\cdot \exp_q\left(nM^{(q)}(\gamma,p)
\right)\;.
\end{align*}
On the other hand, if $w>2t$, then
$\mu_q(n,t,w)=0$.
\end{lemma}

Next, $M^{(q)}$ is also
the lower bound on the (exponent of the) 
intersection size, not only for balls,
but also for two spheres of the same radius.

\begin{lemma} 
\label{lem:nu_t lower bound_general}
Let $q\ge 2$, and $0\le w,t\le n$ be integers
such that $t\le (1-1/q)n$.
Let us write $t=p n$ and $w=\gamma n$. 
If $w>2t$ or $q=2$ and $w$ is odd, then
$\nu_q(n,t,w)=0$. In all other cases,
\begin{align*}
\nu_q(n,t,w)\ge \exp_q\left(n\left[M^{(q)}(\gamma,p)-o(1)\right]
\right)\;,
\end{align*}
where $o(1)$ denotes a function
that goes to $0$ as $n$ goes to infinity,
uniformly for all $0\le p\le 1-1/q$
and $0\le\gamma\le\min(2p,1)$ (possibly depending on $q$).
\end{lemma}

We will also need some other properties
of $M^{(q)}$ and $F^{(q)}$. First, $M^{(q)}$ is continuous.

\begin{lemma}
\label{lem:M_q_continuity}
For every $q\ge 2$, the function $M^{(q)}(\gamma,p)$ is continuous on its domain $\{(\gamma,p): 0\le p\le1-1/q, 0\le \gamma\le\min(2p,1)\}$.
In particular, it is uniformly continuous.
\end{lemma}

As $M^{(q)}(\gamma, p)$ is related to the size
of the intersection of Hamming balls of radius $pn$,
it is easy to show that it is increasing in $p$.
What is more, the function
$F^{(q)}(\gamma, p)=H_q(p)-M^{(q)}(\gamma, p)$
is decreasing in $p$.

\begin{lemma}
\label{lem:fq-decreasing}
For $q\ge 2$ and $0<\gamma\le 1$, the function
$p\mapsto F^{(q)}(\gamma,p)$
is strictly decreasing for $\gamma/2\le p\le 1-1/q$.
\end{lemma}

\subsection{Block error probability bound}

The values $\mu_q(n,t,w)$ play a role
in the bound on the error probability on the
symmetric channel due to Poltyrev~\cite{poltyrev2002bounds}.
A convenient version of this bound is
given by Pathegama and Barg
in~\cite[Proposition 5]{pathegama2023smoothing}. 
Here we give a version adapted to our purposes.
The proof is deferred to \Cref{apx: A1}.

\begin{proposition}
\label{prop:poltyrev}
    Let $\cC\subseteq\Sigma^n$ be a $q$-ary
    code with weight
    distribution $(A_w(x))_x$,
    $0<p\le 1-1/q$, $1\le w_0\le n$,
    and $t=p n+n^{\theta}$
    (not necessarily an integer), where $\theta\in(1/2,1)$. Then, 
    for every $x\in\cC$,
    \begin{align}
    \label{eq:22}
        P_{B}(\cC,\qSC_p, x)&\le 2\exp\left(-2n^{2\theta-1}\right)+\sum_{w=1}^{w_0}A_w(x)Z(\qSC_p)^w\\
        &\qquad+\left(\frac{(1-p)(q-1)}{p}\right)^{n^\theta}q^{-nH_q(p)}\sum_{w=w_0+1}^{n}A_w(x)\mu_q(n,t,w)\;.
    \label{eq:23}
    \end{align}
\end{proposition}

\Cref{prop:poltyrev}
and \Cref{lem:qary_balls_intersection} lead
to an asymptotic bound on the error
probability that we will apply several
times later on.

\begin{proposition}
\label{prop:qsc-vanishing-pb}
Let $\{\cC_n\}_n$ be a family of $q$-ary codes
with distance $d(\cC_n)\ge\omega(\log(nL))$
for some $L=L(n)$.
Let $\alpha>0$ and $0<p<1-1/q$.

Assume that $A_w(x)\le L(n)$ for every $x\in\cC_n$ and
$w<\alpha n$. Furthermore, assume that
for $\alpha n\le w=\gamma n\le \min(2pn+n^{3/4}, n)$ it holds
\begin{align}
\label{eq:24}
\frac{1}{n}\log_q A_{\gamma n}(x)\le G(\gamma)+o(1)
\end{align}
for some continuous function $G(\gamma)$,
and where $o(1)$ denotes a function that goes to $0$ as
$n$ grows, uniformly for all $w$ and $x\in\cC_n$.

If $F^{(q)}(\gamma, p)>G(\gamma)$ for every
$\alpha\le \gamma\le \min(2p, 1)$, then
\begin{align*}
\lim_{n\to\infty} P_B(\cC_n, \qSC_p)=0\;.
\end{align*}
\end{proposition}

\begin{proof}
Let $t\coloneqq pn+n^{3/4}$ and $x\in\cC_n$. 
Below and for the rest of the proof, 
the $o(1)$ and $o(n)$ terms are uniform in $w$ and $x$.
Applying
\Cref{prop:poltyrev} with $w_0=\lceil \alpha n\rceil-1$, 
it holds
\begin{align*}
P_B(\cC_n,\qSC_p,x)
&\le o(1)+\sum_{d(\cC_n)\le w<\alpha n}A_w(x)Z(\qSC_p)^w
+\exp(o(n))q^{-nH_q(p)}\sum_{w\ge\alpha n}A_w(x)\mu_q(n,t,w)\;.
\end{align*}
We will show that each of the two sums above goes to 0,
uniformly in $x$. First, since $p<1-1/q$, it holds
$Z(\qSC_p)<1$. Furthermore, by assumption, we have
$A_w(x)\le L$ for $w<\alpha n$. Hence,
and using $d(\cC_n)\ge\omega(\log(nL))$,
\begin{align*}
\sum_{d(\cC_n)\le w< \alpha n}
A_w(x)Z(\qSC_p)^w\le nL\cdot Z(\qSC_p)^{d(\cC_n)}=o(1)\;.
\end{align*}
We now turn to the second sum.
Let us write $w=\gamma n$
and $p_n\coloneqq t/n=p+n^{-1/4}$. Using the 
assumption~\eqref{eq:24} and \Cref{lem:qary_balls_intersection},
\begin{align*}
\exp(o(n))q^{-nH_q(p)}\sum_{w\ge\alpha n}A_w(x)\mu_q(n,t,w)&\le
\exp(o(n))\sum_{\alpha n\le w\le \min(2t,n)}
\exp_q\left(n\left[
-H_q(p)+G(\gamma)+M^{(q)}(\gamma,p_n)
\right]\right)\\
&\le
\exp_q\left(o(n)-n\left[
\inf_{\alpha\le \gamma\le\min(2p_n,1)}
H_q(p)-G(\gamma)-M^{(q)}(\gamma,p_n)
\right]\right)\;.
\end{align*}
It remains to show that
$H_q(p)-M^{(q)}(\gamma, p_n)-G(\gamma)$ is positive and uniformly
bounded away from $0$.
For $\alpha\le\gamma\le\min(2p, 1)$ this holds
because, by \Cref{lem:M_q_continuity}, we have
$H_q(p)-M^{(q)}(\gamma,p_n)-G(\gamma)
=F^{(q)}(\gamma,p)-G(\gamma)+o(1)$.
But by assumption $F^{(q)}(\gamma, p)-G(\gamma)>0$
is a strictly positive and continuous function
on the compact interval
$\alpha\le \gamma\le \min(2p,1)$, so it holds
$H_q(p)-M^{(q)}(\gamma,p_n)-G(\gamma)\ge \eps>0$
for $n$ large enough.

For $2p<\gamma\le 2p_n$, 
once again by continuity of $M^{(q)}$
and $G$, we have
$H_q(p)-M^{(q)}(\gamma, p_n)-G(\gamma)
=F^{(q)}(2p,p)-G(2p)+o(1)\ge \eps>0$
for large enough $n$. All in all, we showed
that $P_B(\cC_n,\qSC_p,x)\xrightarrow{n\to\infty}0$
uniformly in $x$, therefore also
$P_B(\cC_n,\qSC_p)\xrightarrow{n\to\infty} 0$.
\end{proof}

\subsection{Proof of \Cref{thm:list-vs-qsc-main}}
Let $\{\cC_n\}_n$ be a family of $q$-ary codes satisfying
the assumptions of the theorem. 
Our plan is to apply \Cref{prop:qsc-vanishing-pb}
with the function
$G(\gamma)= F^{(q)}(\gamma, p)$ and $\alpha=p'$.
Indeed, by $(p,L)$-list decodability,
it holds $A_w(x)\le L$ for $w< p'n$. 

On the other hand,
let $t\coloneqq \lfloor p n\rfloor$.
Note that the condition
$d(\cC_n)=\omega(\log(nL))$ ensures that
$L\le\exp(o(n))$.
By \Cref{thm:list-double-counting} applied
for $t_1=t_2=t$, 
and using \Cref{lem:M_q_continuity,lem:nu_t lower bound_general},
for $w=\gamma n\le \min(2t,n)$ it holds\footnote{
For $q=2$ and odd $w$, \eqref{eq:28} is not valid since
$\nu_2(n,t,w)=0$. This can be fixed by considering
$\cC'\subseteq\Sigma^{n+1}$ such that
$\cC\coloneqq\{(x,1):x\in\cC\}$. The code $\cC'$ is list decodable
up to $t$ errors and list size $L$. Therefore, for odd $w\le 2t-1$,
\begin{align*}
    A_w(x)=A'_{w+1}(x,0)\le
    \frac{\binom{n+1}{t}L}{\nu_2(n+1,t,w+1)}
    \le \exp_2\left((n+1)\left[
    h\left(\frac{t}{n+1}\right)+o(1)-M\left(\frac{w+1}{n+1},\frac{t}{n+1}\right)
    \right]\right)\;,
\end{align*}
which is also uniformly bounded by $\exp_2(nF(\gamma,p)+o(n))$.
}
\begin{align}
\label{eq:28}
A_w(x)\le\frac{\binom{n}{t}(q-1)^t L}{\nu_q(n,t,w)}
\le \exp_q\left(n\left(H_q(t/n)+o(1)-M^{(q)}(\gamma, t/n)\right)\right)
\le \exp_q\left(n\left(F^{(q)}(\gamma, p)+o(1)\right)\right)
\;,
\end{align}
so indeed for $w=\gamma n\le 2p'n+n^{3/4}$
it holds $w\le 2t$ and
\begin{align*}
\frac{1}{n}\log_q A_{\gamma n}(x)\le F^{(q)}(\gamma, p)+o(1)=G(\gamma)+o(1)\;,
\end{align*}
uniformly in $x$ and $\gamma$.
Finally, by \Cref{lem:fq-decreasing} it holds
that for $p'\le\gamma\le \min(2p',1)$ we have
\begin{align*}
    F^{(q)}(\gamma, p')-G(\gamma)
    =F^{(q)}(\gamma,p')-F^{(q)}(\gamma,p)>0\;.
\end{align*}
Therefore, by \Cref{prop:qsc-vanishing-pb}
we have $\lim_{n\to\infty}P_B(\cC_n,\qSC_{p'})=0$.
\qed

\section{Symmetric channel decoding
from minimum distance and erasure channel}
\label{sec:samorodnitsky}

\subsection{Weight distribution bounds}

The main technical tool in this section is
a clean bound on the weight distribution of a linear code 
in terms of its conditional entropy on the erasure channel.
The bound is taken from~\cite{abawonse2025generalized} and is a generalization of the binary bound 
from~\cite{hkazla2021codes}.
\begin{theorem}[Theorem 39 in \cite{abawonse2025generalized}]
     \label{prop:weight_distribution_bound}
     Let $n\ge1$, $\cC\subseteq\bF^n_q$ a linear code
     and let $0\le\lambda\le 1$. Let $X$ be a uniformly chosen codeword in $\cC$ and $Y$ be the output of transmitting $X$ over the channel $\qEC_\lambda$. Then, the weight distribution of $\cC$ satisfies
     \begin{align*}
         \sum_{i=0}^nA_i\left(\frac{q^\lambda-1}{q-1}\right)^i\le 2^{H(X|Y)}\;,
     \end{align*}
     where $H(X|Y)$ is the conditional Shannon entropy 
     (with base-2 logarithm) of 
     $X$ given $Y$.
\end{theorem}
One corollary of \Cref{prop:weight_distribution_bound} is a bound on the block error probability of decoding $\cC$ on a variety of channels.
\begin{corollary}[Corollary 44 in \cite{abawonse2025generalized}]
    \label{cor:block_error_bound_W}
    Let $n\ge1$, $0<\lambda\le 1$ and  $\cC\subseteq\bF^n_q$ a linear code with minimum distance $d$. Let $X$ be a uniformly chosen codeword in $\cC$ and $Y$ be the output of transmitting $X$ over the channel $\qEC_\lambda$. Let $W:\mathbb{F}_q\to\mathcal{Y}$ be a channel, and assume that $c\coloneqq Z(W)\frac{q-1}{q^{\lambda}-1}<1$. Then, it holds,
    \begin{align*}
        P_B(\cC,W)\le\frac{c^d}{1-c}\cdot2^{H(X|Y)}\;.
    \end{align*}
\end{corollary}

Our proofs apply~\Cref{prop:weight_distribution_bound} in a
black box fashion, and our results give a generalization 
of \Cref{cor:block_error_bound_W}.

\subsection{Further properties of
$M^{(q)}$ and $F^{(q)}$}

The function $M^{(q)}(\gamma,\delta)$
has more properties that will be useful
in this section. First, it is concave
as a function of $\gamma$.

\begin{lemma}
\label{clm:M_gamma_concavity}
For fixed $q\ge 2$ and $0\le p\le 1-1/q$,
the function $\gamma\mapsto - M^{(q)}(\gamma,p)$ is a convex function on $[0,\min(2p,1)]$.
\end{lemma}

Second, its partial derivative with respect
to $\gamma$ evaluated at $\gamma=0$
is equal to the logarithm of the Bhattacharayya coefficient
$Z(\qSC_p)$. We will use this property in analyzing
the limiting behavior of $p_*^{(q)}(\lambda,\delta)$ as $\delta\to 0$.

\begin{lemma}  
\label{lem:m-derivative-at-0}
Let $q\ge 2$, $0<p\le 1-\frac{1}{q}$, then it holds
\begin{align*}
     \frac{\partial  M^{(q)}(\gamma,p)}{\partial \gamma}\bigg\rvert_{\gamma=0}&=\log_q Z(\qSC_p)=\log_q\left(\frac{q-2}{q-1}p+2\sqrt{\frac{p(1-p)}{q-1}}\right)\;.
\end{align*}
\end{lemma}

The proofs of the lemmas are deferred to \Cref{sec:m-f-proofs}. We also state a claim
which is easy to check from the definitions.

\begin{claim}
\label{cl:f-is-zero-for-gamma-zero}
For every $q\ge 2$ and $0\le p\le 1-1/q$,
it holds $F^{(q)}(0,p)=0$.
\end{claim}

\subsection{Proof of \Cref{Thm:1}}

Let $\{C_n\}_n$ be a family of $q$-ary linear
codes with relative distances $\delta$
and vanishing bit error probability
on $\qEC_\lambda$. We wish to apply
\Cref{prop:qsc-vanishing-pb}
for $\alpha=\delta$ and
$G(\gamma)= \gamma\log_q\left(\frac{q-1}{q^{\lambda}-1}\right)$. Since the code
is linear, we can assume without loss of generality that the transmitted codeword is $x=0^n$ and bound 
$A_w=A_w(0^n)$.

Indeed, for $w<\delta n$ it holds
$A_w=0$. For $w\ge\delta n$, by
\Cref{prop:weight_distribution_bound},
\begin{align*}
\frac{1}{n}\log_q A_{\gamma n}
\le \gamma \log_q\left(\frac{q-1}{q^{\lambda}-1}\right)+\log_q 2\cdot \frac{H(X|Y)}{n}\le
G(\gamma)+o(1)\;,
\end{align*}
where the last step follows since
on the erasure channel the
vanishing bit error probability
implies that $H(X|Y)=o(n)$
(see \Cref{cl:pbit-implies-entropy}).

Let $p_*\coloneqq p^{(q)}_*(\lambda,\delta)$
and recall that $p<p_*$. If $p<\delta/2$,
then the code can correct $p$ fraction of
errors in the worst case, so
$P_B(\cC_n,\qSC_p)\to 0$ follows
by a standard argument without any assumptions
on the erasure channel performance.
Otherwise $F^{(q)}(\delta,p)$ is well defined
and, by the definition
of $p_*$, it holds $F^{(q)}(\delta, p)>G(\delta)$. However, by~\Cref{clm:M_gamma_concavity}, the function
$\gamma\mapsto -M^{(q)}(\gamma,p)$ is
convex, and therefore also
\begin{align*}
W(\gamma)\coloneqq
F^{(q)}(\gamma,p)-G(\gamma)
=H_q(p)-M^{(q)}(\gamma, p)-G(\gamma)
\end{align*}
is convex as a sum convex functions.
Since $W(0)=0$ by \Cref{cl:f-is-zero-for-gamma-zero} and we just established $W(\delta)>0$,
by convexity it also holds $W(\gamma)>0$
for $\gamma\ge\delta$. Therefore,
the assumptions of \Cref{prop:qsc-vanishing-pb}
hold and it follows that
$\lim_{n\to\infty}P_B(\cC_n,\qSC_p)=0$.
\qed

\subsection{Illustration and limiting behavior for $\delta\to 0$}

The following  
\Cref{fig:function_F},
\Cref{fig:delta_for_fixed_alpha},
\Cref{fig:delta_for_fixed_lambda},
and \Cref{fig:q_ary_delta_for_fixed_lambda}
illustrate various aspects of
functions $F^{(q)}(\gamma, p)$ and
$p_*^{(q)}(\lambda,\delta)$.

\begin{figure}[ht]
    \centering
    \includegraphics[width=0.7
    \linewidth]{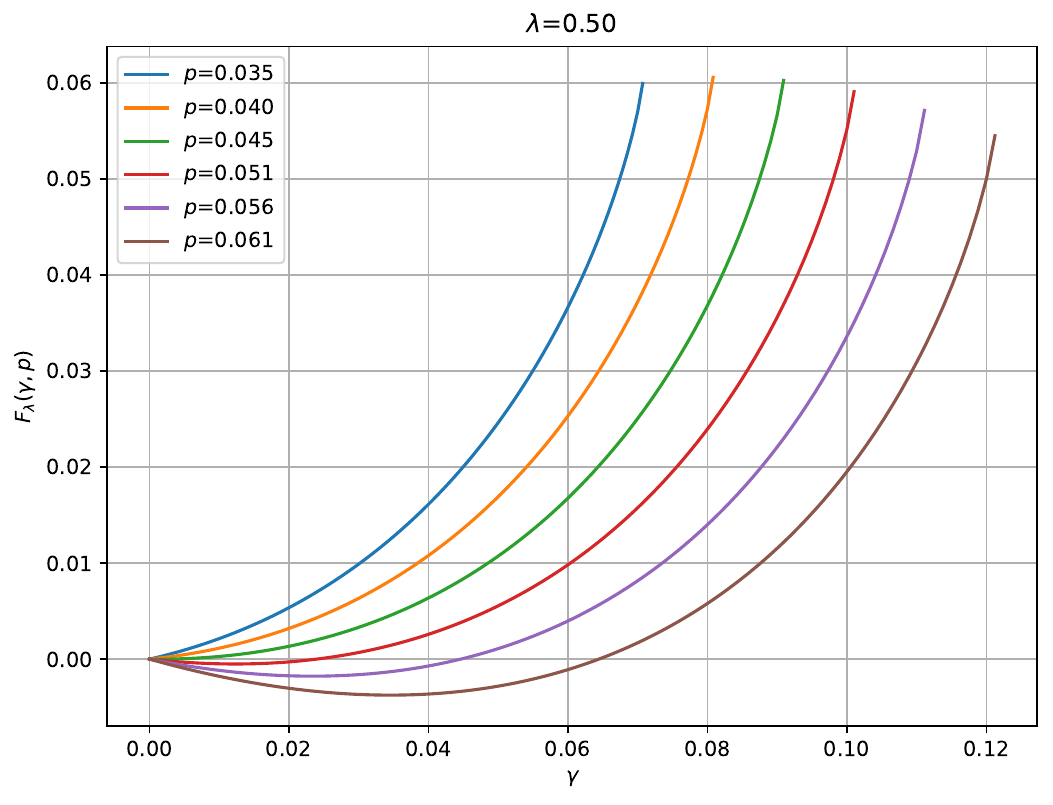}
    \caption{
    Binary case $q=2$, plot of
    $F_\lambda(\gamma, p)\coloneqq F(\gamma,p)+\gamma\log\left(2^\lambda-1\right)$
    as a function of $\gamma$,
    for $\lambda=0.5$ and some values of
    $p$. Our results imply vanishing
    error probability on $\BSC_p$
    whenever  $F_\lambda(\delta, p)>0$,
    assuming relative distance $\delta$
    and vanishing error probability
    on $\BEC_\lambda$.
    }
    \label{fig:function_F}
\end{figure}

\begin{figure}[ht]
    \centering
    \includegraphics[width=0.7\linewidth]{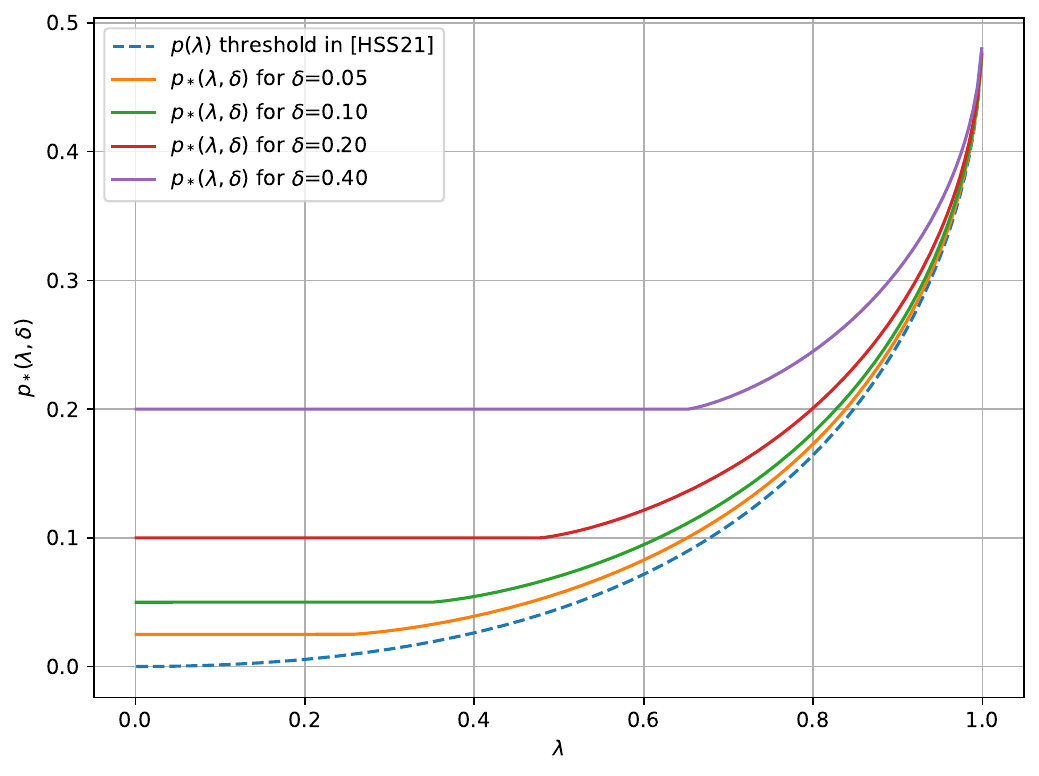}
    \caption{
    Binary case $q=2$, plot of
    $p_*(\lambda,\delta)$ as a function of the BEC erasure probability $\lambda$, for fixed relative distances $\delta\in\{0.05, 0.1, 0.2,0.4\}$. For comparison, we plot the threshold $p(\lambda)=\frac{1}{2}-\sqrt{2^{\lambda-1}(1-2^{\lambda-1})}$
    from~\cite{hkazla2021codes} which
    does not require the assumption of $\delta>0$. In the intervals where the graphs are constant,
    our bound matches the trivial
    bound $p_*(\lambda,\delta)\ge \delta/2$.}
    \label{fig:delta_for_fixed_alpha}
\end{figure}

\begin{figure}[!ht]
    \centering
    \includegraphics[width=0.7\linewidth]{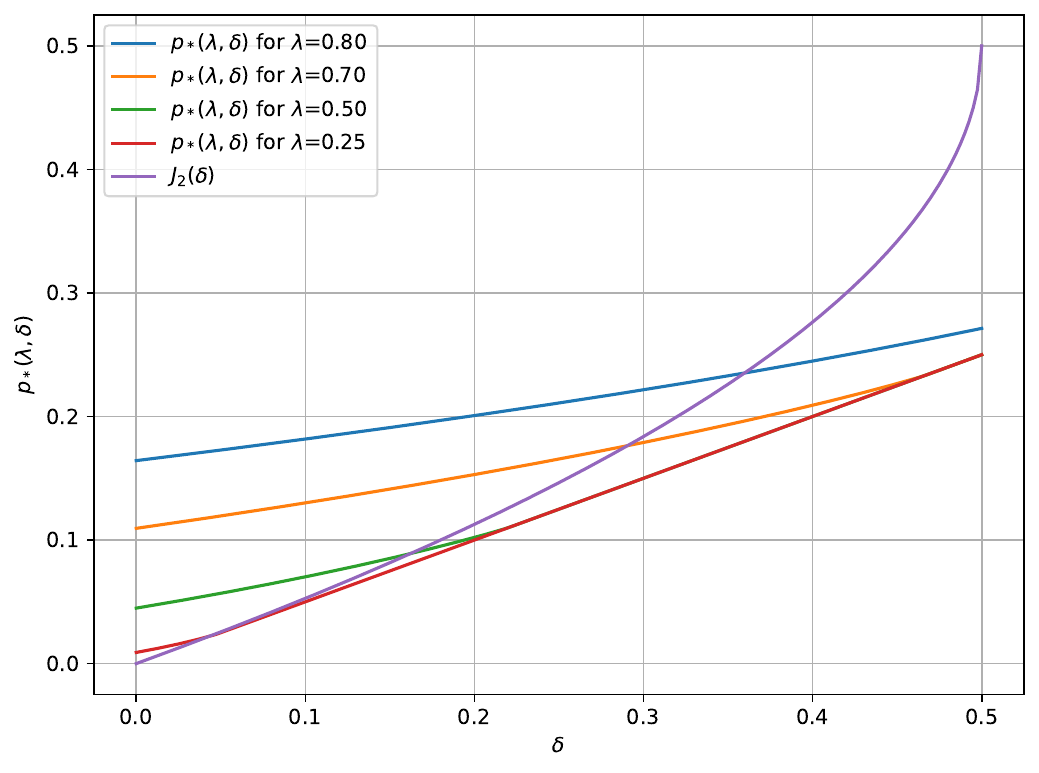}
    \caption{
    Binary case $q=2$,
    $p_*(\lambda,\delta)$ as a function of the relative distance $\delta$, for fixed erasure probabilities $\lambda\in\{0.25, 0.5, 0.7, 0.8\}$. 
    Since a code family with relative distance $\delta$
    has list decoding radius $J_2(\delta)$,
    by \Cref{thm:list-vs-qsc-main}
    it also has vanishing symmetric channel error probability for $p<J_2(\delta)$. Our bounds improve
    upon this result under the additional BEC assumption whenever $p_*(\lambda,\delta)>J_2(\delta)$.
    }
    \label{fig:delta_for_fixed_lambda}
\end{figure}

\Cref{Thm:1} establishes that 
linear codes with relative
distance $\delta$ which perform well on the erasure channel $\qEC_\lambda$ also achieve reliable decoding on the symmetric channel $\qSC_p$, for all $p<p_*^{(q)}(\lambda,\delta)$. 
On the other hand, applying \Cref{thm:original}, a general linear code
that decodes errors on $\qEC_\lambda$ also decodes errors
on $\qSC_p$ as long as
\begin{align}
\label{eq:13}
Z(\qSC
_p)=\frac{q-2}{q-1}p+2\sqrt{\frac{(1-p)p}{q-1}}<\frac{q^{\lambda}-1}{q-1}\;.
\end{align}
\Cref{Thm:1} refines this picture under the additional assumption that the code family has linear minimum distance, i.e. $\delta(\cC_n)\ge\delta$.
We will now argue that the threshold $p^{(q)}_*(\lambda,\delta)$ is increasing in $\delta$, and in the limit $\delta\to0^+$ matches the earlier threshold of~\eqref{eq:13}. 

\begin{lemma}
\label{lem:delta_limit}
Let $q\ge 2$ and $0<\lambda\le1$. Then, $p_*^{(q)}(\lambda,\delta)$ is 
increasing 
as a function of $\delta>0$. Furthermore,
\begin{align*}
    p_*^{(q)}(\lambda)\coloneqq \lim_{\delta \to 0^+}p_*^{(q)}(\lambda,\delta)
\end{align*}
satisfies $Z\left(\qSC_{p_*^{(q)}(\lambda)}\right)=\frac{q^\lambda-1}{q-1}$,
in particular
$p_*^{(2)}(\lambda)=\frac{1}{2}-\sqrt{2^{\lambda-1}(1-2^{\lambda-1})}$.
\end{lemma}

\begin{proof}
    Let $q\ge 2$ and $0<\lambda\le 1$. Let
    \begin{align*}
    W(\gamma,p)&\coloneqq F^{(q)}(\gamma, p)+\gamma\log_q\left(\frac{q^{\lambda}-1}{q-1}\right)
    =H_q(p)-M^{(q)}(\gamma,p)+\gamma\log_q\left(\frac{q^{\lambda}-1}{q-1}\right)\;.
    \end{align*}
    Recall that $p_{*}^{(q)}(\lambda,\delta)
    =\inf\{\delta/2\le p\le 1-1/q: W(\delta, p)\le 0\}$. 
    Let $p_0\coloneqq p_{*}^{(q)}(\lambda,\delta_0)$, 
    $p_1\coloneqq p_{*}^{(q)}(\lambda,\delta_1)$
    for some $0<\delta_0<\delta_1\le 1$. Note that $\delta_0/2<\delta_1/2\le p_1$, so the values
    of $F^{(q)}(\delta_0,p_1)$ and $W(\delta_0,p_1)$
    are well defined.

    If $W(\delta_0,p_1)\le 0$, then
    by definition of $p_0$ it holds $p_0\le p_1$.
    On the other hand, let us show that $W(\delta_0,p_1)>0$
    is not possible.
    Note that $\gamma\mapsto W(\gamma,p_1)$ is convex as a sum of 
    $-M^{(q)}(\gamma,p_1)$ (which is convex by \Cref{clm:M_gamma_concavity}) and a linear function.
    Furthermore, $W(0, p_1)=0$, so 
    $W(\delta_0,p_1)>0$ implies
    $W(\delta_1,p_1)> 0$.
    But since $W$ is continuous, that contradicts
    the definition of $p_1$. Therefore,
    the function $\delta\mapsto p_*^{(q)}(\lambda,\delta)$
    is increasing.

    Now let $p_*\coloneqq p_*^{(q)}(\lambda)$ and let $p_0$ satisfy $Z(\qSC_{p_0})=\frac{q^{\lambda}-1}{q-1}$. We are going to
    show that $p_*=p_0$.
    For $p<p_0$, recall from \Cref{lem:m-derivative-at-0} that
    \begin{align*}
    \frac{\partial}{\partial\gamma}W(\gamma,p)\bigg\rvert_{\gamma=0}
    &=-\log_qZ(\qSC_p)+\log_q\left(\frac{q^{\lambda}-1}{q-1}\right)
    =: c>0
    \end{align*}
Since the function $\gamma\mapsto W(\gamma,p)$ is convex, then for all $\delta>0$, it holds $W(\delta,p)\ge c\delta>0$, thus by definition of $p_*^{(q)}(\lambda,\delta)$ and 
\Cref{lem:fq-decreasing}, we have $p_*(\lambda,\delta)\ge p$,
and taking $\delta\to 0$ and then $p\to p_0$ it follows $p_*\ge p$
and $p_*\ge p_0$.

On the other hand, for $p>p_0$, we have 
$\frac{\partial}{\partial\gamma}W(\gamma,p)\bigg\rvert_{\gamma=0}
=:-c<0$.
Therefore, there exists $\delta_0>0$ such that 
$W(\delta,p)\le 0$ for every $0<\delta\le\delta_0$,
which implies
$p_*^{(q)}(\lambda,\delta)\le p$ for $0<\delta\le\delta_0$,
and as before
$p_*\le p$ and $p_*\le p_0$.
\end{proof}

Figures~\ref{fig:delta_for_fixed_alpha} illustrates the improvement of $p^{(2)}_*(\lambda,\delta)$ over the thresholds in~\cite{hkazla2021codes} for different parameter choices.

\begin{figure}[!ht]
    \centering
    \includegraphics[width=0.7\linewidth]{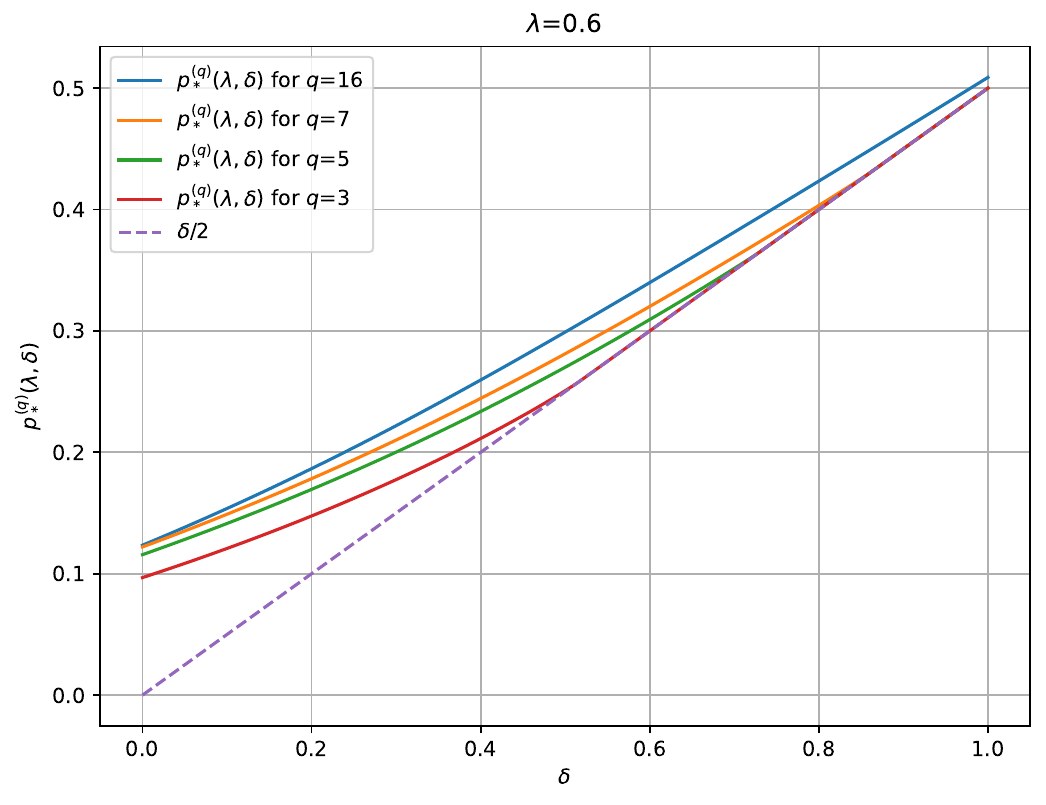}
    \caption{The threshold $p_*^{(q)}(\lambda,\delta)$ as a function of the relative distance $\delta$, for fixed erasure probability $\lambda=0.6$ and $q\in\{3,5,7,16\}$. We also provide the trivial bound $\delta/2$.
    (As in \Cref{fig:delta_for_fixed_lambda}
    our bound can be compared with the Johnson
    radius $J_q(\delta)$. We omit this in the
    interest of visual clarity.)
    }
    \label{fig:q_ary_delta_for_fixed_lambda}
\end{figure}

\subsection{Binary dual bound}
\label{sec:dual}

There is another bound for the weight distribution
in terms of the conditional entropy on the
erasure channel of the \emph{dual code}
$\cC^\perp\coloneqq \{x:\sum_{i=1}^n x_iy_i=0 \ \forall y\in\cC\}$.
A similar result to \Cref{Thm:1} can be obtained
by replacing \Cref{prop:weight_distribution_bound}
with this dual bound. For simplicity,
in this section we restrict ourselves to $q=2$.

\begin{definition}
Let $0\le\lambda\le R< 1$. Let
\begin{align*}
    G^\perp_{\lambda,R}(\gamma)&\coloneqq \begin{cases}R-\lambda-\min(\gamma,1-\gamma)\log(2^{1-\lambda}-1)& \text{if}\quad 0\le \min(\gamma,1-\gamma)<1-2^{\lambda-1},\\
    h(\gamma)-(1-R)&\text{otherwise.}\end{cases}
\end{align*}
Then, for $0\le\delta\le 1$, let
\begin{align*}
    p^{\perp}_*(\lambda,R,\delta)\coloneqq\inf\left\{\delta/2\le p\le 1/2: F(\delta,p)\le G^\perp_{\lambda,R}(\delta)\right\}\;.
\end{align*}
\end{definition}

\begin{theorem}
\label{Thm:main_dual}
   Let $0<\delta,\lambda<1$, let $\cC=\{\mathcal{C}_n\}_{n}$ be a family of binary linear codes of rate $R=\lim_n R(\cC_n)$ with $\delta(\cC_n)\ge\delta$, and such that       $\underset{n\to\infty}{\lim}P_b(\cC^{\perp}_n,\BEC_{\lambda})=0.$  Then,  
   \begin{align*}
       \underset{n\to\infty}{\lim}P_B(\cC_n,\BSC_{p})=0,\ \text{for every}\ p<p^{\perp}_{*}(\lambda,R,\delta).
   \end{align*}
\end{theorem}

The function $G^\perp_{\lambda,R}$ is defined
only  for $\lambda\le R$. This is without
loss of generality, since $\lambda\le R$
follows from the Shannon's coding theorem
if the dual code family 
has vanishing bit error probability
on the $\BEC_\lambda$.
Since $F(\delta,1/2)=0$
and $G_{\lambda,R}^\perp(\delta)\ge 0$
(see \Cref{prop:negative_G}), the
value of $p^\perp_*(\lambda,R,\delta)$ is well defined.

We will need two prerequisites before we
prove \Cref{Thm:main_dual}. First,
let us state the dual weight distribution
bound due to~\cite{samorodnitsky2019upper}:
\begin{proposition}[Corollary 40 in \cite{abawonse2025generalized}]
\label{lem:qary_Sam}
    Let $\cC$ be a linear code over $\mathbb{F}_q^{n}$, with weight distribution $(A_0,\ldots,A_n)$. Let $0<\lambda\le 1$ and $\tau\coloneqq \frac{q^\lambda-1}{q-1}$.
    \begin{itemize}
        \item[1)] Let $X$ be a uniformly chosen random element in $\cC$, $Y$ be the output of transmitting $X$ through $\qEC_\lambda$. Then, for every $0\le w\le n$, it holds $A_w\le \tau^{-w} 2^{H(X|Y)}$.
        \item[2)] Similarly, let $X^{\perp}$ be a uniformly chosen random element in $\cC^{\perp}$, and $Y^{\perp}$ be the output of transmitting $X^{\perp}$ through $\qEC_\lambda$. Let $\gamma=\frac{w}{n}$, then, for every $0\le w\le n$, it holds
        \begin{align*}
            A_w\le|\cC|2^{H(X^{\perp}|Y^{\perp})}\begin{cases}(1-\tau)^{-\gamma n}(1+(q-1)\tau)^{-(1-\gamma)n}&\quad\text{if}\ \gamma<\frac{q-1}{q}(1-\tau),\\ q^{-n(1-H_q(\gamma))}&\quad\text{if}\ \frac{q-1}{q}(1-\tau)\le \gamma\le \frac{q-1}{q}(1+\tau),
            \\
            (1+\tau)^{-\gamma n}
            (1-(q-1)\tau)^{-(1-\gamma)n}&\quad\text{if}\ \gamma>\frac{q-1}{q}(1+\tau). \end{cases} 
        \end{align*}
        \end{itemize}
\end{proposition}

Second, we need to check that 
$G^\perp_{\lambda,R}$ is a concave function.

\begin{lemma}
    \label{lem:F_dual_convex}
          The function $\gamma\mapsto G^{\perp}_{\lambda,R}(\gamma)$ is a concave function for $0\le\gamma\le 1$.
\end{lemma}
\begin{proof}
  Let $G(\gamma)\coloneqq G^\perp_{\lambda,R}(\gamma)$,
  $m\coloneqq\min(\gamma,1-\gamma)$ and $m_0:=1-2^{\lambda-1}$. 
  Since $R-\lambda-m_0\log(2^{1-\lambda}-1)=h(m_0)-1-R$,
  the function $G(\gamma)$ is continuous on $[0,1]$. 
  Also 
\begin{align*}
    G'(\gamma)=\begin{cases}
            \mp\log(2^{1-\lambda}-1) \quad\ \text{if}\ 0\le m<m_0\\h'(\gamma)\quad\ \text{otherwise}\;,
        \end{cases}
\end{align*}
in particular one checks that
$G'^{-}(m_0)=-\log(2^{1-\lambda}-1)=h'(m_0)=G'^{+}(m_0)$, and by a symmetric argument
$G'^-(1-m_0)=G'^+(1-m_0)$. Since $G$ has continuous first derivative,
to conclude that it is concave
it is sufficient to check that
$G''(\gamma)\le 0$ for every $m\ne m_0$.
However, this is clear since
    \begin{align*}
        G''(\gamma)=\begin{cases}
            0 &\text{if }m<m_0\\
            h''(\gamma)&\text{if } m>m_0\;.
        \end{cases}
    \end{align*}
    In either case, $G''(\gamma)\le 0$, so $G$ is concave.
\end{proof}

\begin{proof}[Proof of \Cref{Thm:main_dual}]
    The proof follows a similar outline as \Cref{Thm:1}. 
    We assume can that $\delta\le 2p$, as otherwise
    $P_B(\cC_n,\BEC_p)\to 0$ follows from
    $\delta(\cC_n)\ge \delta$. As we discussed,
    the Shannon's coding theorem 
    and $\lim_n P_b(\cC_n^\perp,\BEC_\lambda)=0$
    imply $R\ge\lambda$. 

    We will apply \Cref{prop:qsc-vanishing-pb}
    for $\alpha=\delta$ and 
    $G(\gamma)\coloneqq G^\perp_{\lambda,R}(\gamma)$. Since the code is linear,
    we can bound $A_w=A_w(0^n)$ w.l.o.g.
    For $0<w<\delta n$ we have $A_w=0$
    by the distance assumption.

    For $w=\gamma n\ge\delta n$, by \Cref{lem:qary_Sam}
    it holds
    \begin{align*}
    \frac{1}{n}\log A_{\gamma n}
    &\le G(\gamma)+(R_n-R)+\frac{H(X^\perp|Y^\perp)}{n}\;.
    \end{align*}
    The term $R_n-R$ is $o(1)$ by assumption
    and $H(X^\perp|Y^\perp)=o(n)$ due
    to $P_b(\cC^\perp_n,\BEC_\lambda)\to 0$
    (see~\Cref{cl:pbit-implies-entropy}).

    Finally, since
    $p<p_*^\perp(\lambda,R,\delta)$,
    we have $F(\delta,p)>G(\delta)$.
    Furthermore, 
    $\gamma\mapsto F(\gamma, p)-G(\gamma)
    =h(p)-M(\gamma,p)-G(\gamma)$
    is convex by \Cref{clm:M_gamma_concavity}
    and \Cref{lem:F_dual_convex}.
    Since $F(0,p)-G(0)=\lambda-R\le 0$,
    from $F(\delta,p)-G(\delta)>0$
    it follows $F(\gamma,p)-G(\gamma)>0$
    for every $\delta\le \gamma\le 2p$.
    Accordingly, by \Cref{prop:qsc-vanishing-pb},
    it follows that
    $P_B(\cC_n,\BSC_p)\to 0$.
\end{proof}

\paragraph{Comparison with the primal bound.}
Consider a family of binary linear codes with 
relative minimum distance $\delta$ and
$\lim_n R(\cC_n)=R$
that achieves capacity on the BEC, that is 
$\lim_{n} P_b(\cC_n,\BEC_{\lambda})=0$ for every $\lambda<1-R$.
Then, by \Cref{Thm:1} (and continuity) it follows
that $\{\cC_n\}_n$ has vanishing block error probability
for $\BSC_p$ for $p<p_*(1-R,\delta)$.
At the same time, it is known that $\{\cC_n\}_n$ achieves capacity on the BEC
(under bit-MAP decoding)
if and only if the dual code family $\{\cC_n^\perp\}_n$ does so,
see, e.g., \cite[Remark 12]{kudekar2016reed}.
Therefore, by, \Cref{Thm:main_dual}, $\{C_n\}_n$ decodes errors
on the BSC for $p<p_*^\perp(R,R,\delta)$.

Accordingly, it is natural to ask which of the two bounds
is the better one.
It is straightforward to check from the
definitions that
for $0\le\delta\le 1-2^{R-1}$ it holds
\begin{align*}
p_*(1-R,\delta)=p^\perp_*(R,R,\delta)
=\inf\left\{\delta/2\le p\le 1/2:
F(\delta,p)\le -\delta\log(2^{1-R}-1)\right\}\;.
\end{align*}

Furthermore,
it can be checked that the known asymptotic bounds between minimum 
distance and rate 
exclude the existence of
binary code families with
rate $R$ and $\delta>1-2^{R-1}$. Therefore,
$p_*(1-R,\delta)=p_*^\perp(R,R,\delta)$ for all feasible
combinations of $\delta$ and $R$, and
\Cref{Thm:main_dual} does not improve upon \Cref{Thm:1}
for codes that achieve capacity on the BEC.
\Cref{Thm:main_dual} remains potentially applicable in scenarios
where $R>\lambda$.
\Cref{fig:delta_and_dual_for_fixed_lambda} provides an illustration of the dependence
of $p_*^\perp$ on $R$.

\begin{figure}[!ht]
    \centering
    \includegraphics[width=0.6\linewidth]{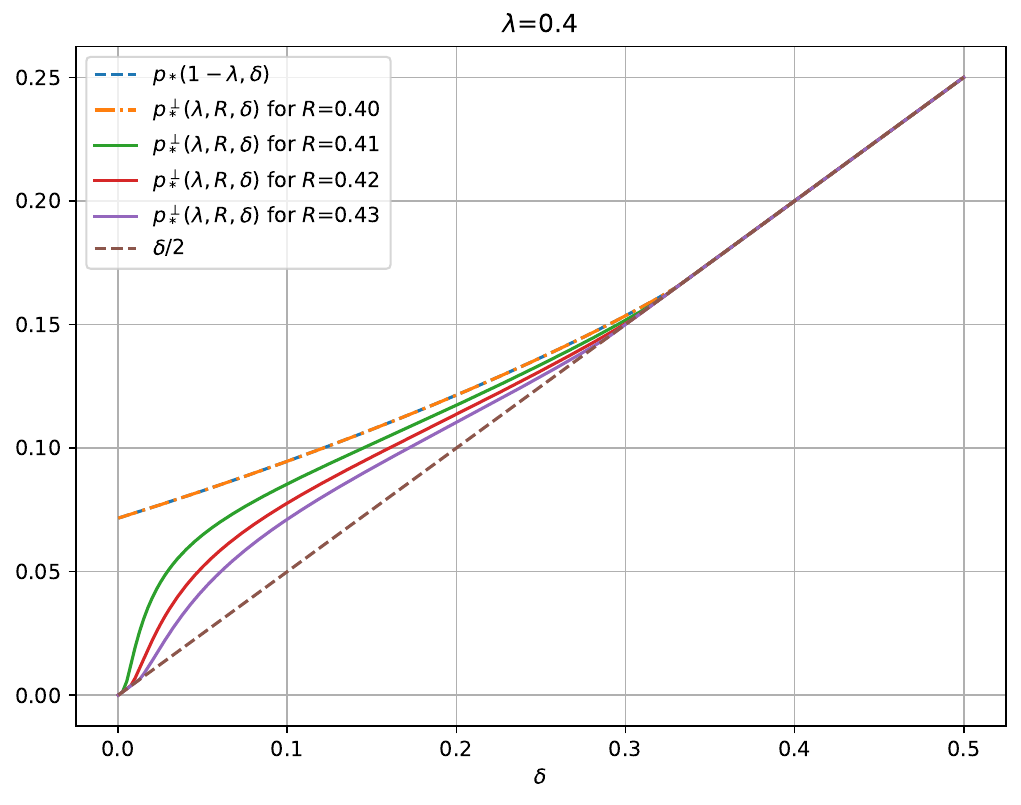}
    \caption{$p_*(0.6,\delta)$  and $p^{\perp}_*(0.4,R,\delta)$, as a function of the relative distance $\delta$, for fixed rates $R\in\{0.4, 0.41,, 0.42, 0.43\}$.}
    \label{fig:delta_and_dual_for_fixed_lambda}
\end{figure}

\section{Proof of \Cref{thm:unconditional_bound}}
\label{sec:unconditional_bound} 

To establish \Cref{thm:unconditional_bound}, we need two preliminary claims.
\begin{claim}[Exercise 7.17 in \cite{guruswami2012essential}] \label{clm:erasure_list_decocable} Let $\cC\subseteq\bF^n_q$ be a code with relative distance at least $0\le \delta\le 1-1/q$. Then, for every $\varepsilon>0$, $\cC$ is $\left(\left(\frac{q}{q-1}-\varepsilon\right)\delta,\frac{q}{(q-1)\varepsilon}\right)$-erasure list decodable.
\end{claim} 
\begin{proof}
For $\delta=0$ the statement is trivial, so assume $\delta>0$.
Let $\varepsilon>0$, and $T\subseteq [n]$ such that
    $|T|=(1-\rho)n$ for some $\rho\le \left(\frac{q}{q-1}-\eps\right)\delta$. Let $y\in\Sigma^T$. We 
    want to show that 
    \begin{align*}
        \big|\{c\in \cC:c_T=y\}\big|\le\frac{q}{(q-1)\varepsilon}\;.
    \end{align*} 
    Consider the truncated code $\cC'\coloneqq \{c_{\bar{T}}: c\in\cC, c_T=y\}$, where $\bar{T}=[n]\backslash T$. By construction, $|\cC'|=|\{c\in \cC:c_T=y\}|$. Furthermore, the minimum distance 
    of $\cC'$ satisfies $d(\cC')\ge d(\cC)\ge\delta n$, and since
    $\cC'\subseteq\Sigma^{\bar{T}}$ its block length is $n'=|\bar{T}|=\rho n$. Then, we have 
    \begin{align*}
        \frac{q-1}{q}n'
        =\frac{q-1}{q}\rho n\le
        \delta n-\frac{q-1}{q}\eps\delta n
        <\delta n\le d(\cC')\;.
    \end{align*}
The Plotkin bound (\cite{Plo60}, see also~\cite[Theorem 4.4.1]{guruswami2012essential})
states that if a $q$-ary code $\cC$ with block length $n$ 
satisfies $\frac{q-1}{q}n<d\le d(\cC)$
then $|\cC|\le\frac{qd}{qd-(q-1)n}$. Applying this to $\cC'$,
   \begin{align*}
       |\cC'|&\le\frac{q\delta n}{q\delta n-(q-1)\rho n}\le \frac{q}{(q-1)\varepsilon}\;.\qedhere
   \end{align*}
\end{proof}

In~\cite{pernice2025list} it is proved that a 
$(p,L)$-list decodable code family has vanishing error
probability on the $\qEC_\lambda$ for $\lambda<p$. However,
their proof can be adapted to establish the same conclusion
assuming $(p,L)$-\emph{erasure} list decodability instead of
list decodability.

\begin{lemma}\label{clm:erasure_reliability}
Let $\{\cC_n\}_n$ be a family of $q$-ary codes and $0\le\rho\le 1$.
If $\cC_n$ is $(\rho,L)$-erasure list decodable for some $L=L(n)$
satisfying $d(\cC_n)\ge \omega(\log L)$, then for every $\rho'<\rho$ it
holds
\begin{align*}
    \lim_{n\to\infty}P_{B}(\cC_n,\qEC_{\rho'})=0\;.
\end{align*}
\end{lemma}
\begin{proof}
Take some $\rho'<\rho_0<\rho$ and write $\rho'=(1-\eps)\rho_0$.
 Consider the following two-step procedure.
First, sample a random set $\cE_0\subseteq [n]$,
such that every $1\le i\le n$ is placed in $\cE_0$ independently
with probability $\rho_0$. Then, conditioned on $\cE_0$,
choose $\cE'\subseteq\cE_0$ such that every $i\in\cE_0$ is
placed in $\cE'$ independently with probability $1-\eps$.
Note that the overall distribution of $\cE'$ is that of the
erasure pattern on the $\qEC_{\rho'}$.

Let $x\in\cC_n$.
Given $\cE\subseteq[n]$, let 
$S(\cE)\coloneqq\{y\in\cC: y_{\bar{\cE}}=x_{\bar{\cE}}\}$, where
$\bar{\cE}$ denotes the complement of $\cE$. The set $S(\cE)$
contains the codewords that cannot be distinguished from $x$ given
the erasure pattern $\cE$. Since if $|S(\cE)|=1$, then
$S(\cE)=\{x\}$ and the MAP decoder
is guaranteed to decode $x$ correctly, 
to show vanishing error probability on $\qEC_{\rho'}$
it is sufficient to upper
bound the probability that $|S(\cE')|>1$.

If $|\cE_0|\le \rho n$, then by the erasure list decoding property
it holds $|S(\cE_0)|\le L$. Fix some $\cE_0$ such that
$|S(\cE_0)|\le L$ and let us write $S(\cE_0)=\{x,y_1,\ldots,y_k\}$ for
some $k<L$. Observe that, conditioned on $\cE_0$, 
for every $1\le j\le k$ it holds
\begin{align*}
    \Pr\left[y_j\in \cE'\;\vert\; \cE_0\right]
    &=(1-\eps)^{\Delta(x,y_j)}\le (1-\eps)^{d(\cC_n)}\;.
\end{align*}
Applying the union bound, it holds $\Pr[|S(\cE')|>1\;\vert\;\cE_0]
\le L(1-\eps)^{d(\cC_n)}$. Since $\Pr[|\cE_0|>\rho n]=o(1)$ by the Hoeffding
bound, we have
\begin{align*}
\Pr\left[|S(\cE')|>1\right]
&\le \Pr\left[|\cE_0|>\rho n\right]
+\Pr\left[|S(\cE')>1|\;\vert\; |\cE_0|\le\rho n\right]\\
&\le o(1)+L(1-\eps)^{d(\cC_n)}=o(1)\;,
\end{align*}
where in the last step we invoke the assumption
$d(\cC_n)\ge \omega(\log L)$.
\end{proof}

\begin{proof}[Proof of~\Cref{thm:unconditional_bound}]
    Let $\lambda_*\coloneqq \frac{q}{q-1}\delta$
    and $p_*\coloneqq p_*^{(q)}(\lambda_*,\delta)$.
    Let $p<p_*$ and $\{\cC_n\}_n$ a $q$-ary linear code family with
    $\delta(\cC_n)\ge \delta$.
    If $p<\delta/2$, then
    $P_B(\cC_n,\qSC_p)\to 0$ follows trivially
    from $\delta(\cC_n)\ge\delta$.
    Otherwise,
    from \Cref{def:q_ary_functions}, the condition $p<p_*$
    implies $F^{(q)}(\delta,p)>\delta\log_q\left(\frac{q-1}{q^{\lambda_*}-1}\right)$. By continuity, there exists
    some $\lambda<\lambda_*$ such that
    $F^{(q)}(\delta,p)>\delta\log_q\left(\frac{q-1}{q^{\lambda}-1}\right)$.
    Since $F^{(q)}$ is continuous and decreasing in $p$
    (\Cref{lem:M_q_continuity} and \Cref{lem:fq-decreasing}),
    it follows that $p<p_*^{(q)}(\lambda, \delta)$.
    
    By \Cref{clm:erasure_list_decocable} applied for $\eps=\frac{1}{n}$, $\cC_n$ is $\left(\frac{q\delta}{q-1}-\frac{\delta}{n},\frac{qn}{(q-1)}\right)$-erasure list decodable. Then, by \Cref{clm:erasure_reliability}, it follows that 
    $\lim_{n\to\infty}P_{B}(\cC_n,\qEC_{\lambda'})=0$ for every $\lambda'<\lambda_*$, in particular
    $P_{B}(\cC_n,\qEC_{\lambda})=0$ for our previous choice of $\lambda$.
    Since $p<p_*^{(q)}(\lambda,\delta)$
    (and the vanishing block error probability implies
    vanishing bit error probability), from
    \Cref{Thm:1} it follows $\lim_{n\to\infty}P_B(\cC_n,\qSC_p)=0$.

    Since we showed vanishing block error probability for
    every linear code family with $\delta(\cC_n)\ge\delta$ and every
    $p<p_*$, indeed it holds $\pLsym(q,\delta)\ge p_*$.
\end{proof}

\section{Proofs of the properties of
$M^{(q)}$ and $F^{(q)}$}
\label{sec:m-f-proofs}

\subsection{``Almost closed form''
for $M^{(q)}$
}

While the function $M^{(q)}$ does not
seem to have a simple closed form for
$q\ge 3$, we can give an ``almost''
closed form for it. This is useful not only for
numerically computing the values of
$M^{(q)}$, but also in several of the proofs
that follow.

\begin{lemma}
\label{lem:almost-closed-form}
Let $0\le p\le 1-1/q$ and $0\le\gamma\le\min(2p,1)$.
Let $C_q\coloneqq\frac{(q-2)^2}{4(q-1)}$,
$\beta(\zeta)\coloneqq\frac{1}{1+C_q(\frac{1-\zeta}{\zeta})^2}$.
Then,
        \begin{align*}
            M^{(q)}(\gamma,p)=\gamma\tilde{H}_q(\zeta)+(1-\gamma)H_q(\beta(\zeta))
        \end{align*}
        for the unique $\zeta$ such that $\gamma\frac{\zeta}{2}+(1-\gamma)\beta(\zeta)=p-\frac{\gamma}{2}$.
\end{lemma}

Since $\beta(\zeta)$ is a strictly increasing
function of $\zeta$, it is easy to see that the
value of $\zeta$ is unique. It also means that
$\zeta$, and consequently $M^{(q)}(\gamma,p)$
can be efficiently computed via binary search.
(In principle one can also determine the value
of $\zeta$ by reducing the formula
for $\beta(\zeta)$ to a polynomial equation
of degree 3. However, this seems more cumbersome.)

\begin{proof}
Let $D\coloneqq D(\gamma,p)$
and $T(\zeta,\beta)\coloneqq
\gamma\Htilde_q(\zeta)+(1-\gamma)H_q(\beta)$.
Consider a pair of values $(\zeta,\beta)\in D$
that realizes the supremum in the definition
of $M^{(q)}(\gamma,p)$.
First, let us argue that we can assume that
$\gamma\zeta/2+(1-\gamma)\beta=p-\gamma/2$.
Indeed, note that
$\gamma(1-2/q)/2+(1-\gamma)(1-1/q)=1-1/q-\gamma/2
\ge p-\gamma/2$. Therefore, if
$\gamma\zeta/2+(1-\gamma)\beta<p-\gamma/2$,
then either $\zeta<1-2/q$ or $\beta<1-1/q$.
If $\zeta<1-2/q$, then for some small $\eps>0$
it holds $(\zeta+\eps,\beta)\in D$
and $T(\zeta+\eps,\beta)\ge T(\zeta,\beta)$.
A symmetric argument applies if $\beta<1-1/q$.

To sum up, it holds that
\begin{align}
\label{eq:M_achive_at_sphere}
    M^{(q)}(\gamma,p)=\sup_{\substack{0\le \zeta,\beta\le 1\\
    \gamma\zeta/2+(1-\gamma)\beta=p-\gamma/2}} T(\zeta,\beta)\;.
\end{align}
Let us apply the Lagrange multiplier method 
to $T(\zeta,\beta)$
with the constraint
$\gamma\zeta/2+(1-\gamma)\beta=p-\gamma/2$.
The Lagrangian is 
    \begin{align*}
        \mathcal{L}(\zeta,\beta,\lambda')= \gamma\tilde{H}_q(\zeta)+(1-\gamma)H_q(\beta)+\lambda'\Big(p-\frac{\gamma}{2}-\gamma\frac{\zeta}{2}-(1-\gamma)\beta\Big),
    \end{align*}
and we have 
    \begin{align*}
        \nabla \mathcal{L}=0\iff\begin{cases}
            \log_q(\frac{1-\zeta}{\zeta})+\log_q(\frac{q-2}{2})=\lambda'/2 &\\ \log_q(\frac{1-\beta}{\beta})+\log_q(q-1)=\lambda' &\\ \gamma\frac{\zeta}{2}+(1-\gamma)\beta=p-\frac{\gamma}{2}\;.
        \end{cases}
    \end{align*}
    By substituting $\lambda'$ from the second equation into the first equation, we get $\log\left((\frac{1-\zeta}{\zeta})^2\times\frac{\beta}{1-\beta}\right)=\log(\frac{1}{C_q})$ i.e., $\beta=\frac{1}{1+C_q(\frac{1-\zeta}{\zeta})^2}$.
    
    Since the function
    $\zeta\mapsto \gamma\zeta/2+(1-\gamma)\beta(\zeta)$
    is strictly increasing from 0 to $1$ for
    $0\le\zeta\le 1$, indeed there is unique
    $\zeta$ satisfying 
    $\gamma\zeta/2+(1-\gamma)\beta(\zeta)=p-\gamma/2$. By the Lagrange multiplier
    calculation above, this point must achieve
    the supremum of $T$. (The only other possibilities
    have $\zeta\in\{0,1\}$ or
    $\beta\in\{0,1\}$, but these
    are excluded by the fact that
    $\lim_{\zeta\to 0^+}\Htilde'_q(\zeta)
    =\lim_{\beta\to 0^+}H'_q(\beta)=\infty$,
    as well as 
    $\lim_{\zeta\to 1^-}\Htilde'_q(\zeta)
    =\lim_{\beta\to 1^-}H'_q(\beta)=-\infty$.)
\end{proof}

\subsection{Proof of \Cref{lem:qary_balls_intersection}}

If $w>2t$, then clearly $\mu_q(n,t,w)=0$,
so we focuse on the case $w\le 2t$.
Below we proceed in two cases,
binary and $q\ge 3$.

\paragraph{Binary case $q=2$.}
   Let $B\coloneqq
   B(0^n,t)\cap B(1^w0^{n-w},t)$, so $\mu_2(n,t,w)=|B|$. Assume w.l.o.g.~that
   $\Sigma=\mathbb{F}_2$.
   For a given $y\in\mathbb{F}_2^n$, let $a(y):=|\{1\le j\le w: y_j=1\}|$, $b(y):=|\{w<j\le n: y_j=1\}|$. Then,
   \begin{align*}
       y\in B \iff \begin{cases}a(y)+b(y)\le t \\ w-a(y)+b(y)\le t\;.\end{cases}
   \end{align*}
   So, 
   \begin{align*}    \mu_2(n,t,w)=\sum_{a=0}^{w}\sum_{b=0}^{n-w} \mathds{1}[a+b\le t]\mathds{1}[w-a+b\le t]\binom{w}{a} \binom{n - w}{b}.
   \end{align*}
   
   Recall that for every $0\le\zeta\le \gamma\le1\ \text{we have}\ \binom{\gamma n}{\zeta n} \le2^{\gamma n h(\frac{\zeta}{\gamma})}$. Then, substituting $w=\gamma n, a=\zeta \gamma n, b=\beta(1-\gamma) n$, 
   as the sum has at most $n^2$ nonzero terms, we have 
   \begin{align*}
       \mu_2(n,t,w)&=\underset{\substack{0\le a\le \gamma n\\0\le b\le (1-\gamma) n}}{\sum} 
       \mathds{1}[a+b\le pn]\mathds{1}[\gamma n-a+b\le pn]
       \binom{\gamma n}{\zeta \gamma n} \binom{(1-\gamma)n}{\beta(1-\gamma) n}
       \\&\le n^2
       \sup_{(\zeta,\beta)\in D} \exp_2\left(n[\gamma h(\zeta)+(1-\gamma)h(\beta)]\right)\;,
   \end{align*}
where $D\coloneqq\Big\{(\zeta,\beta)\in [0,1]^2:\gamma\zeta + (1-\gamma)\beta\le p,\ \gamma(1-\zeta) + (1-\gamma)\beta\le p\Big\}$. To conclude the proof,
it is sufficient to show that

\begin{align*}
\sup_{(\zeta,\beta)\in D}
\gamma h(\zeta)+(1-\gamma)h(\beta)
=\gamma+(1-\gamma)h\left(
\frac{p-\gamma/2}{1-\gamma}
\right)
=M(\gamma,p)\;.
\end{align*}
To see that, first note
that 
$\left(1/2,\frac{p-\gamma/2}{1-\gamma}\right)\in D$,
therefore $\mathrm{LHS}\ge\mathrm{RHS}$.
On the other hand, if $(\zeta,\beta)\in D$,
then $(1/2,\beta)\in D$, and if
$(1/2,\beta)\in D$, then
$\beta\le\frac{p-\gamma/2}{1-\gamma}
\le 1/2$. So, indeed
\begin{align*}
\gamma h(\zeta)
+(1-\gamma)h(\beta)
&\le \gamma+(1-\gamma)h\left(\frac{p-\gamma/2}{1-\gamma}\right)=M(\gamma,p)\;.
\pushQED{\qed}\qedhere\popQED
\end{align*}

\paragraph{Large alphabets $q\ge 3$.}
As in the binary case, assume $\Sigma=\mathbb{Z}_q$,
and let $B\coloneqq B(0^n,t)\cap B(1^{w}0^{n-w},t)$, so $\mu_q(n,t,w)=|B|$. For a given $y\in\mathbb{F}_q^n$, let $a(y)\coloneqq |\{0\le j\le w: y_j=0\}|$, $c(y)\coloneqq|\{0\le j\le w: y_j=1\}|$ and $b(y)\coloneqq|\{w<j\le n: y_j=0\}|$.
   By definition,
\begin{align*}
       y\in B \iff \begin{cases}n-a(y)-b(y)\le t \\ n-c(y)-b(y)\le t\end{cases}\ ,
   \end{align*}
so
\begin{align*}
   \mu_q(n,t,w)&=\sum_{a=0}^{w}\sum_{c=0}^{w-a}
    {\sum_{b=0}^{n-w}} \Bigg(
    \mathds{1}\left[n-a-b\le t\right]
    \mathds{1}\left[n-c-b\le t\right] \cdot \\
    &\qquad\qquad \cdot \binom{w}{a}\binom{w-a}{c}(q-2)^{w-a-c} \binom{n - w}{b}(q-1)^{n-w-b}\Bigg)
    \\&=\sum_{s=0}^{w}\sum_{a=0}^{s}
    {\sum_{b=0}^{n-w}} 
    \Bigg(
    \mathds{1}\left[\max(n-a-b,n-s+a-b)\le t\right]
    \cdot \\
    &\qquad\qquad \cdot
    \binom{w}{s}\binom{s}{a}(q-2)^{w-s}
    \binom{n - w}{b}(q-1)^{n-w-b}\Bigg)
    \\&\le\sum_{s=0}^{w}\sum_{b=0}^{n-w}
    \mathds{1}\left[n-s/2-b\le t\right]
    2^s\binom{w}{s}(q-2)^{w-s}
    \binom{n - w}{b}(q-1)^{n-w-b}\;.
  \end{align*}
In the last step we used the fact that
$n-s/2-b\le\max(n-a-b,n-s+a-b)$.
Substituting $b=\beta(1-\gamma)n$, and $s=2\zeta\gamma n$, we get 
\begin{align*}
    \mu_q(n,t,w)&\le\sum_{s=0}^{w}
    {\sum_{b=0}^{n-w}}
    \mathds{1}\left[n-s/2-b\le t\right] 2^{2\zeta\gamma n}
    \binom{\gamma n}{2\zeta\gamma n}
    (q-2)^{(1-2\zeta)\gamma n} \binom{(1-\gamma)n}{\beta(1-\gamma)n}(q-1)^{(1-\beta)(1-\gamma)n}
    \\&\le n^2\exp_2\left(n
    \underset{\substack{(\zeta,\beta)\in[0,1/2]\times[0,1]\\1- \zeta\gamma - \beta(1-\gamma) \le p}}{\sup}\gamma \left[2\zeta+h(2\zeta)+(1-2\zeta)\log(q-2)\right]+(1-\gamma)\left[h(\beta)+(1-\beta)\log(q-1)\right]
    \right)\;.
\end{align*}
We have $2\zeta+h(2\zeta)+(1-2\zeta)\log(q-2)=-2\zeta\log(\zeta)-(1-2\zeta)\log(\frac{1-2\zeta}{q-2})=\log(q)\tilde{H}_q(1-2\zeta)$, 
and $h(\beta)+(1-\beta)\log(q-1)=\log(q)H_q(1-\beta)$. Since
\begin{align*}
\underset{\substack{(\zeta,\beta)\in[0,1/2]\times[0,1]\\1- \zeta\gamma - \beta(1-\gamma) \le p}}{\sup}\gamma\tilde{H}_q(1-2\zeta)+(1-\gamma)H_q(1-\beta)
&=\underset{\substack{(\zeta,\beta)\in[0,1]^2\\1-(\frac{1-\zeta}{2})\gamma - (1-\beta)(1-\gamma) \le p}}{\sup}\gamma\tilde{H}_q(\zeta)+(1-\gamma)H_q(\beta)\\
&=\sup_{(\zeta,\beta)\in D(\gamma,p)}
\gamma\tilde{H}_q(\zeta)+(1-\gamma)H_q(\beta)
=M^{(q)}(\gamma,p)\;,
\end{align*}
indeed it holds
\begin{align*}
    \mu_q(n,t,w)&
    \le n^2 \exp_2{\left(n \log q
    \sup_{(\zeta,\beta)\in D(\gamma,p)}
\gamma\tilde{H}_q(\zeta)+(1-\gamma)H_q(\beta)
    \right)}\\&=
    n^2\exp_q\left(nM^{(q)}(\gamma,p)\right)\;.
    \pushQED{\qed}\qedhere\popQED
\end{align*}

\subsection{Proof of \Cref{lem:nu_t lower bound_general}}

The structure of the proof is similar
as in \Cref{lem:qary_balls_intersection}.
For  $w > 2t$, the result is obvious.
For $w\le 2t$, we consider $q=2$
and $q\ge 3$ separately.

\paragraph{Binary case $q=2$.}
For any $y\in\bF^n_q$, let $a(y):=|\{1\le i\le w : y_i=0\}|$ and $b(y):=|\{w<i\le n: y_i= 1\}|$. We have $\Delta(0^n,y)=w-a(y)+b(y)$ and $\Delta(1^w0^{n-w})=a(y)+b(y)$,
thus $y\in\bF^n_q$ satisfies $\Delta(0^n,y)=t$ and $\Delta(1^w0^{n-w},y)=t$
if and only if
\begin{align*}
    w-a(y)+b(y)=a(y)+b(y)=t\;.
\end{align*}
These equations have the unique solution
$a(y)=w/2$ and $b(y)=t-w/2$.
So, for odd $w$ there is no integer solution
and $\nu_2(n,t,w)=0$. On the other hand,
for even $w$ we can check
that $0\le w/2\le w$ and
$0\le t-w/2\le (n-w)/2\le n-w$,
and therefore
\begin{align*}
\nu_2(n,t,w)&=
\binom{w}{w/2}\binom{n-w}{t-w/2}
=\binom{\gamma n}{\frac{\gamma/2}n}
\binom{(1-\gamma)n}{(p-\gamma/2)n}
\ge\frac{1}{4n}
\exp_2\left(n\left[\gamma+(1-\gamma)
h\left(\frac{p-\gamma/2}{1-\gamma}\right)\right]\right)\\
&=\frac{1}{4n}\exp_2\left(nM(\gamma,p)\right)\;.
\end{align*}
where in the last step we used a bound on binomial
coefficients $\binom{n}{\gamma n}\ge 
2^{h(\gamma)n}/(2\sqrt{n})$, which is valid
for all $n\ge 1$ and $0\le\gamma n\le n$.
(Also in the special case $\gamma=1,p=1/2$
and even $n$ it is easy to check
that $\nu_2(n,n/2,n)=\binom{n}{n/2}\ge\frac{1}{2\sqrt{n}}
2^n=\frac{1}{2\sqrt{n}}\exp_2(nM(1,1/2))$.)
\qed

\paragraph{Larger alphabets $q\ge 3$.}
Recall that we write $w=\gamma n$ and
$t=p n$.
  Let $S\coloneqq S(0^n,t)\cap S(1^w 0^{n-w},t)$,
  so that $\nu_q(n,t,w)=|S|$.
  For $y\in \Sigma^n$, let $a(y):=|\{1\le i\le w : y_i=0\}|$, $c(y):=|\{ 1\le i\le w : y_i= 1\}|$ and $b(y):=|\{w<i\le n: y_i= 0\}|$,
  as well as $s(y)\coloneqq a(y)+c(y)$.
  Then $y\in S$ if and only if
\begin{align*}
        n-a(y)-b(y)=n-c(y)-b(y)=t\;,
\end{align*}
  which implies $a(y)=c(y)$.
  Let
  \begin{align*}
  D_0\coloneqq 
  \left\{(s,b): 0\le s\le w, 0\le b\le n-w,
  s\text{ even}, n-s/2-b=t\right\}\;,
  \end{align*}
  and let us write $s=(1-\zeta)\gamma n$ and
  $b=(1-\beta)(1-\gamma)n$.
  Let $T(\gamma,\zeta,\beta)\coloneqq\gamma \Htilde_q(\zeta)+(1-\gamma)H_q(\beta)$.
  We will use the binomial coefficient bound
  $\binom{n}{\gamma n}\ge 
  \frac{\exp_2(h(\gamma)n)}{2\sqrt{n}}$. Accordingly,
  \begin{align}
  \nu_q(n,t,w)&=\sum_{(s,b)\in D_0} \binom{w}{s}\binom{s}{s/2}
  (q-2)^{w-s}\binom{n-w}{b}(q-1)^{n-w-b}\nonumber\\
  &=\sum_{(s,b)\in D_0} \binom{\gamma n}{\zeta\gamma n}
  \binom{(1-\zeta)\gamma n}{(1/2-\zeta/2)\gamma n}
  (q-2)^{\gamma\zeta n}\binom{(1-\gamma)n}{\beta(1-\gamma)n}
  (q-1)^{(1-\gamma)\beta n}\nonumber\\
  &\ge \frac{1}{8n^{3/2}}\sum_{(s,b)\in D_0}
  \exp_2\left(n\left[
  \gamma h(\zeta)+\gamma(1-\zeta)+
  \gamma \zeta\log (q-2)
  +(1-\gamma)h(\beta)+(1-\gamma)\beta\log(q-1)
  \right]\right)\nonumber\\
  &= \frac{1}{8n^{3/2}}\sum_{(s,b)\in D_0}
  \exp_q(nT(\gamma,\zeta,\beta))
  \ge \frac{1}{8n^{3/2}}\exp_q\left(n
  \max_{(s,b)\in D_0}T(\gamma,\zeta,\beta)
  \right)\;.
  \label{eq:27}
  \end{align}
  Let 
  \begin{align*}
  D&\coloneqq\left\{(\zeta,\beta):0\le\zeta,\beta\le 1,\
  \gamma\frac{1+\zeta}{2}+(1-\gamma)\beta=p\right\}\;.
  \end{align*}
  From \Cref{lem:almost-closed-form} (see also~\eqref{eq:M_achive_at_sphere}) it follows that
  \begin{align*}
  M^{(q)}(\gamma,p)=\sup_{(\zeta,\beta)\in D}T(\gamma,\zeta,\beta)\;,
  \end{align*}
  and in light of~\eqref{eq:27} it remains to show
  $|M^{(q)}(\gamma,p)-\max_{(s,b)\in D_0}T(\gamma,\zeta,\beta)|\le o(1)$,
  uniformly in $\gamma$ and $p$.
  To that end, let $(\zeta^*,\beta^*)\in D$ be the point
  given by \Cref{lem:almost-closed-form} such that
  $M^{(q)}(\gamma, p)=T(\gamma,\zeta^*,\beta^*)$.
  Let $s^*\coloneqq (1-\zeta^*)\gamma n$ and 
  $b^*\coloneqq (1-\beta^*)(1-\gamma)n$.
  Note that $0\le s^*\le w$, $0\le b^*\le n-w$,
  and that $b^*=n-t-s^*/2$.
  
  Now let $s_0\coloneqq 2\lfloor s^*/2 \rfloor$
  and $b_0\coloneqq n-t-s_0/2$ and write
  $s_0=(1-\zeta_0)\gamma n$ and $b_0=(1-\beta_0)(1-\gamma)n$.
  Let us argue that $(s_0,b_0)\in D_0$. Indeed, by construction 
  $s_0$ is an even integer and $n-s_0/2-b_0=t$.
  Furthermore, $0\le s_0\le s^*\le w$. Finally, observe that
  $b_0=n-t-s_0/2=n-t-\lfloor s^*/2\rfloor=\lceil n-t-s^*/2\rceil
  =\lceil b^*\rceil$.
  Since $n-w$ is integer, from $0\le b^*\le n-w$ it follows
  $0\le b^* \le b_0\le n-w$. Hence, $(s_0,b_0)\in D_0$.

  Our objective is to show that
  $|T(\gamma,\zeta^*,\beta^*)-
  \max_{(s,b)\in D_0} T(\gamma,\zeta,\beta)|$ goes to $0$,
  uniformly in $\gamma$ and $p$. It is easy to see that
  if $(s,b)\in D_0$, then the respective $(\zeta,\beta)\in D$,
  so $\max_{(s,b)\in D_0} T(\gamma,\zeta,\beta)\le T(\gamma,\zeta^*,\beta^*)$. Therefore it suffices to bound
  \begin{align}
  \label{eq:25}
  T(\gamma,\zeta^*,\beta^*)-\max_{(s,b)\in D_0} T(\gamma,\zeta,\beta)
  &\le T(\gamma,\zeta^*,\beta^*)-T(\gamma,\zeta_0,\beta_0)
  \le o(1)\;.
  \end{align}
  To that end, we will change variables one more time. Let
  \begin{align*}
  T'(\gamma,u,v)\coloneqq \gamma \Htilde_q\left(\frac{u}{\gamma}\right)
  +(1-\gamma)H_q\left(\frac{v}{1-\gamma}\right)\;,
  \end{align*}
  where we extend $T'$ for $\gamma\in\{0,1\}$ assuming that
  $0\cdot \Htilde_q(u/0)=0\cdot H_q(v/0)=0$. It is easy
  to check that the function $T'$ is continuous on the compact set
  $0\le \gamma,u,v\le 1$, and therefore uniformly continuous.
  Our objective from~\eqref{eq:25} reduces to showing
  \begin{align}
  \label{eq:26}
  \left| T'\big(\gamma,\gamma\zeta^*,(1-\gamma)\beta^*\big)
  -T'\big(\gamma,\gamma\zeta_0,(1-\gamma)\beta_0\big)\right|\le o(1)\;.
  \end{align}
  However, we have $|\gamma \zeta^*n-\gamma \zeta_0 n|=|s_0-s^*|\le 2$,
  hence $|\gamma \zeta^*-\gamma \zeta_0|\le 2/n$, which goes to
  0 uniformly. Similarly, it holds
  $|(1-\gamma)\beta^* n-(1-\gamma)\beta_0 n|=|b-b^*|\le 1$,
  hence $|(1-\gamma)\beta^*-(1-\beta)\beta_0|\le 1/n$.
  Therefore, \eqref{eq:26} holds by the uniform continuity of $T'$.
\qed

\subsection{Proof of \Cref{lem:M_q_continuity}}
In the case $q=2$ the statement is obvious
from the formula~\eqref{eq:11},
and it is easy to check that the function
is continuous with $M(1,1/2)=1$.
Hence, assume $q\ge 3$.
Let $0\le p\le1-1/q, 0\le \gamma\le\min(2p,1)$. From \Cref{lem:almost-closed-form}, there exits a unique $\zeta(\gamma,p)$ satisfying $\gamma\frac{\zeta(\gamma,p)}{2}+(1-\gamma)\beta(\zeta(\gamma,p))=p-\frac{\gamma}{2}$ such that 
\begin{align*}
            M^{(q)}(\gamma,p)=\Phi(\gamma,\zeta(\gamma,p))\;,
        \end{align*}
where $\Phi(\gamma,\zeta)\coloneqq\gamma\tilde{H}_q(\zeta)+(1-\gamma)H_q(\beta(\zeta))$. The function $\zeta\mapsto\beta(\zeta)$ is continuous, thus $\Phi$ is continuous as a sum and composition of continuous functions. So, it is sufficient to show that $\zeta(\gamma,p)$ depends continuously on $(\gamma,p)$. Let's consider the function $\Psi(\gamma,p,\zeta)= \gamma\zeta/2+(1-\gamma)\beta(\zeta)-p+\gamma/2$. $\Psi$ is continuous with respect to all its variables and $\zeta(\gamma,p)$ is the unique solution of the equation $\Psi(\gamma,p,\zeta)=0$. Consider a sequence $\{(\gamma_n,p_n)\}_n$ that converges to $(\gamma,p)$. Let $\zeta_n:=\zeta(\gamma_n,p_n)$, we want to show that $\zeta_n\to\zeta(\gamma,p)$. Let $\{\zeta_{n_k}\}_k$ be a subsequence of $\{\zeta_n\}_n$ that converges to some $\zeta_*$, since $\Psi$ is continuous, we have $\lim_k \Psi(\gamma_{n_k},p_{n_k},\zeta_{n_k})= \Psi(\gamma,p,\zeta_*)$. Moreover, by definition  of $\zeta_{n_k}$, we have $\Psi(\gamma_{n_k},p_{n_k},\zeta_{n_k})=0$, so $\Psi(\gamma,p,\zeta_*)=0$, thus $\zeta_*=\zeta(\gamma,p)$, hence $\zeta_{n_k}\to\zeta(\gamma,p)$.
Since this holds for every convergent
subsequence, it follows
$\zeta_n\to\zeta(\gamma,p)$.

Furthermore, the domain $\{(\gamma,p): 0\le p\le1-1/q, 0\le \gamma\le\min(2p,1)\}$ is compact, therefore $M^{(q)}$ is uniformly continuous.\qed

\subsection{Proof of \Cref{lem:fq-decreasing}}
\paragraph{Binary case $q=2$.}
It is sufficient to show that $\frac{\partial F(\gamma,p)}{\partial p}< 0$
        for $\gamma/2<p<1/2$.
        By direct calculation, $\frac{\partial F(\gamma,p)}{\partial p}=h'(p)-h'\left(\frac{p-\frac{\gamma}{2}}{1-\gamma}\right)$. Since $h''(p)<0$ for all $0<p<1/2$, i.e., $h'$ is decreasing, and for $p<1/2$ and
        $0<\gamma<1$ we have  $\frac{p-\frac{\gamma}{2}}{1-\gamma}< p$, thus $h'(p)<h'\left(\frac{p-\frac{\gamma}{2}}{1-\gamma}\right)$ and therefore $\frac{\partial F_{\lambda}(\gamma,p)}{\partial p}<0$.\qed

\paragraph{Larger alphabets $q\ge 3$.}
Recall that $F^{(q)}(\gamma,p)= H_q(p)-M^{(q)}(\gamma, p)$. As $F^{(q)}$ is continuous,
it suffices to show that $\frac{\partial F^{(q)}(\gamma,p)}{\partial p}< 0$ 
for $\gamma/2< p<1-1/q$,
i.e., $H_q'(p)< \frac{\partial M^{(q)}(\gamma,p)}{\partial p}$. From \Cref{lem:almost-closed-form}, there exists unique $\zeta:=\zeta(p)$ such that $M^{(q)}(\gamma,p)=\gamma\tilde{H}_q(\zeta)+(1-\gamma)H_q(\beta(\zeta))$ with $\gamma\frac{\zeta}{2}+(1-\gamma)\beta(\zeta)=p-\frac{\gamma}{2}$. Thus,
\begin{align*}
    \frac{\partial M^{(q)}(\gamma,p)}{\partial p}=\gamma\zeta'\tilde{H}_q'(\zeta)+(1-\gamma)\zeta'\beta'(\zeta)H_q'(\beta(\zeta))\;,
\end{align*}
where we write $\zeta'=\zeta'(p)$ for short
(note that $\zeta$ is differentiable by the implicit function theorem).
Applying the derivative with respect to $p$ to the constraint, we get $\gamma\frac{\zeta'}{2}+(1-\gamma)\zeta'\beta'(\zeta)=1$ i.e., $(1-\gamma)\zeta'\beta'(\zeta)=1-\gamma\frac{\zeta'}{2}$. Recall that $H_q'(\zeta)=\log_q\left(\frac{(q-1)(1-\zeta)}{\zeta}\right)$, $\tilde{H}_q'(\zeta)=\log_q\left(\frac{(q-2)(1-\zeta)}{2\zeta}\right)$ and $\frac{1-\beta(\zeta)}{\beta(\zeta)}=\frac{(q-2)^2}{4(q-1)}\left(\frac{1-\zeta}{\zeta}\right)^2$. Then, 
\begin{align*}
    \frac{\partial M^{(q)}(\gamma,p)}{\partial p}&=\gamma\zeta'\tilde{H}_q'(\zeta)+\left(1-\frac{\gamma\zeta'}{2}\right)H_q'(\beta(\zeta))
    \\&=\gamma\zeta'\log_q\left(\frac{(q-2)(1-\zeta)}{2\zeta}\right)+\left(1-\frac{\gamma\zeta'}{2}\right)\log_q\left(\frac{(q-1)(1-\beta(\zeta))}{\beta(\zeta)}\right)
    \\&=\gamma\zeta'\log_q\left(\frac{(q-2)(1-\zeta)}{2\zeta}\right)+\frac{2-\gamma\zeta'}{2}\log_q\left(\left(\frac{(q-2)(1-\zeta)}{2\zeta}\right)^2\right)
    \\&=2\log_q\left(\frac{(q-2)(1-\zeta)}{2\zeta}\right)\;.
\end{align*}
So, $\frac{\partial M^{(q)}(\gamma,p)}{\partial p}> H_q'(p)$ if and only if $\frac{(q-2)^2}{4}\left(\frac{1-\zeta}{\zeta}\right)^2>\frac{(q-1)(1-p)}{p}$ i.e., $\frac{1-\beta(\zeta)}{\beta(\zeta)}> \frac{1-p}{p}$, which is equivalent to $p> \beta(\zeta)$. Since  $p=\gamma\frac{1+\zeta}{2}+(1-\gamma)\beta(\zeta)$, 
using $\gamma>0$ it suffices to prove that $\beta(\zeta)<\frac{1+\zeta}{2}$, and we have, for
$0<\zeta<1$,
\begin{align*}
    \beta(\zeta)<\frac{1+\zeta}{2}\iff 
    0< (\zeta-1)(C_q\zeta^2+\zeta^2-C_q)
    \iff
    \zeta^2< \frac{C_q}{1+C_q} \iff \zeta<1-2/q\;.
\end{align*}
Since also $\beta(0)=0<1/2$, it is sufficient that
we show $\zeta<1-2/q$.
However, if $\zeta\ge1-2/q$, as the function $\zeta\mapsto\beta(\zeta)$ is increasing, then $\beta(\zeta)\ge\beta(1-2/q)=\frac{q-1}{q}$, thus $p\ge\gamma\frac{2-2/q}{2}+(1-\gamma)\frac{q-1}{q}=1-1/q$, which contradicts the assumption of $p< 1-1/q$.\qed

\subsection{Proof of \Cref{clm:M_gamma_concavity}}

\paragraph{Binary case $q=2$.}
        Recall that $M(\gamma,p)=\gamma+(1-\gamma)h\Big(\frac{p-\frac{\gamma}{2}}{1-\gamma}\Big)$. We can easily check that the first and second derivatives of $\gamma\mapsto M(\gamma,p)$ are, respectively
        \begin{align*}
            \frac{\partial M(\gamma,p)}{\partial\gamma}=1-h\Big(\frac{p-\frac{\gamma}{2}}{1-\gamma}\Big)
            +\frac{p-1/2}{(1-\gamma)}h'\Big(\frac{p-\frac{\gamma}{2}}{1-\gamma}\Big)
        \end{align*} and
        \begin{align*}
           \frac{\partial^2 M(\gamma,p)}{\partial\gamma^2}=\frac{(p-1/2)^2}{(1-\gamma)^3}h''\Big(\frac{p-\frac{\gamma}{2}}{1-\gamma}\Big)\;,
        \end{align*}
        where $h'$ and $h''$ are respectively the first and second derivative of the binary entropy function. Since $h''(x)\le 0$ for all $0<x<1$, then $\frac{\partial^2 M(\gamma,p)}{\partial^2\gamma}\le 0$, and function $\gamma\mapsto -M(\gamma,p)$ is convex.
        \qed

\paragraph{Larger alphabets $q\ge 3$.}
Let us prove that $-M^{(q)}$ is convex
for $\gamma\notin\{0,1\}$. Then, the convexity
for all $\gamma$ follows since $M^{(q)}$
is continuous.
Recall from \Cref{def:q_ary_functions} that
\begin{align*}
      - M^{(q)}(\gamma,p)= \underset{\substack{\zeta,\beta}}{\inf}-\gamma\tilde{H}_q(\zeta)-(1-\gamma)H_q(\beta)
\end{align*}
    with $(\zeta,\beta)\in[0,1]^2$ such that $\gamma\frac{\zeta}{2}+(1-\gamma)\beta\le p-\frac{\gamma}{2}$. Changing variables $u\coloneqq\gamma\zeta$ and $v\coloneqq(1-\gamma)\beta$, we continue 
    \begin{align*}
         -M^{(q)}(\gamma,p)= \underset{\substack{(u,v)\in [0,\gamma]\times[0,1-\gamma]\\ \frac{u}{2}+v\le p-\frac{\gamma}{2}}}{\inf}-\gamma\tilde{H}_q\left(\frac{u}{\gamma}\right)-(1-\gamma)H_q\left(\frac{v}{1-\gamma}\right).
    \end{align*}
Let $T(u,v,\gamma)\coloneqq-\gamma\tilde{H}_q(\frac{u}{\gamma})-(1-\gamma)H_q(\frac{v}{1-\gamma})$ and $D_T\coloneqq\{(u,v,\gamma): u\in[0,\gamma], v\in [0,1-\gamma], 0<\gamma<1\}$. 

Let us sketch the rest of the proof.
Recall that $-\Htilde_q(u)$ is a convex function.
The function $-\gamma\Htilde_q(u/\gamma)$ is known as the \emph{perspective}
function of $-\Htilde_q$ and it is convex (as a function
of two variables $u$ and $\gamma$) by a standard argument, see \cite[Section 3.2.6]{boyd2004convex}.
Similarly, the function $-(1-\gamma)H_q(v/(1-\gamma))$ is
convex and therefore $T(u,v,\gamma)$ is convex on $D_T$
as a sum of two convex functions.
Finally, $-M^{(q)}$ is convex as the minimum of $T$
(eliminating variables $u$ and $v$)
over a convex subset of $D_T$, see \cite[Section 3.2.5]{boyd2004convex}.

For convenience, let us give a self-contained proof
implementing the sketch above. First, let us argue that $T$
is convex on $D_T$.
Consider $(u_1,v_1,\gamma_1),(u_2,v_2,\gamma_2)\in D_T$
and $0\le\lambda\le 1$.
For $t=(t_1,t_2)$ let's denote $\bar{t}=\lambda t_1+(1-\lambda)t_2$. 
Then,
\begin{align*}
T\left(\bar{u},\bar{v},\bar{\gamma}\right)&=-\bar{\gamma}
\Htilde_q\left(\frac{\bar{u}}{\bar{\gamma}}\right)
-(1-\bar{\gamma})H_q\left(\frac{\bar{v}}{1-\bar{\gamma}}\right)
        \\&=-\bar{\gamma}\Htilde_q\left(\frac{\lambda\gamma_1}{\bar{\gamma}}\frac{u_1}{\gamma_1}+\frac{(1-\lambda)\gamma_2}{\bar{\gamma}}\frac{u_2}{\gamma_2}\right)-(1-\bar{\gamma})H_q\left(\frac{\lambda(1-\gamma_1)}{1-\bar{\gamma}}\frac{v_1}{1-\gamma_1}+\frac{(1-\lambda)(1-\gamma_2)}{1-\bar{\gamma}}\frac{v_2}{1-\gamma_2}\right)
        \\&\le-\lambda\gamma_1\Htilde_q\left(\frac{u_1}{\gamma_1}\right)-(1-\lambda)\gamma_2
\Htilde_q\left(\frac{u_2}{\gamma_2}\right)- \lambda(1-\gamma_1)H_q\left(\frac{v_1}{1-\gamma_1}\right)-(1-\lambda)(1-\gamma_2)H_q\left(\frac{v_2}{1-\gamma_2}\right)
\\&=\lambda T(u_1,v_1,\gamma_1)+ (1-\lambda)T(u_2,v_2,\gamma_2).
\end{align*}
Accordingly, $T$ is also convex on the convex subset
$D\coloneqq D_T\cap\{u/2+v\le p-\gamma/2\}$.
Furthermore, for every $\gamma$,
$D_\gamma\coloneqq\{(u,v):(u,v,\gamma)\in D\}$
is compact and convex,
so the infimum $
\inf_{u,v\in D_\gamma}T(u,v,\gamma)$
is achieved for some $u_\gamma,v_\gamma\in D_\gamma$.

Consider $\gamma_1,\gamma_2$ such that
$-M^{(q)}(\gamma_1,p)=T(u_1,v_1,\gamma_1)$
and $-M^{(q)}(\gamma_2,p)=T(u_2,v_2,\gamma_2)$.
Let $0\le\lambda\le 1$. Then,
 \begin{align*}
     -M^{(q)}(\bar{\gamma},p)&=\inf_{(u,v)\in D_{\bar{\gamma}}}
     T(u,v,\bar{\gamma})
     \le T(\bar{u},\bar{v},\bar{\gamma})
     \le \lambda T(u_1,v_1,\gamma_1)+(1-\lambda)T(u_2,v_2,\gamma_2)
     \\&=-\lambda M^{(q)}(\gamma_1,p)-(1-\lambda)M^{(q)}(\gamma_2,p)\;,
 \end{align*}
so indeed $-M^{(q)}$ is convex as a function of $\gamma$.\qed

\subsection{Proof of \Cref{lem:m-derivative-at-0}}

\paragraph{Binary case $q=2$.}
Let $0<p\le 1/2$, we have $\frac{\partial M(\gamma,p)}{\partial\gamma}=1-h\Big(\frac{p-\frac{\gamma}{2}}{1-\gamma}\Big)+\frac{p-1/2}{(1-\gamma)}h'\Big(\frac{p-\frac{\gamma}{2}}{1-\gamma}\Big)$, hence
\begin{align*}
\frac{\partial M(\gamma,p)}{\partial\gamma}\bigg\rvert_{\gamma=0}&=1-h(p)+(p-1/2)h'(p)
    \\&=1+p\log p+(1-p)\log(1-p)+(p-1/2)\log\left(\frac{1-p}{p}\right)
    \\&=\log\left(2\sqrt{p(1-p)}\right).
    \pushQED{\qed}\qedhere\popQED
\end{align*}

\paragraph{Larger alphabets $q\ge 3$.}
From \Cref{lem:almost-closed-form}, there exists a unique $\zeta=\zeta(\gamma)$ satisfying $\gamma\frac{\zeta}{2}+(1-\gamma)\beta(\zeta)=p-\frac{\gamma}{2}$ (with $\beta(\zeta)=\frac{1}{1+C_q(\frac{1-\zeta}{\zeta})^2}$ and $C_q=\frac{(q-2)^2}{4(q-1)}$) such that $M^{(q)}(\gamma,p)=\gamma\tilde{H}_q(\zeta)+(1-\gamma)H_q(\beta(\zeta))$. Thus
\begin{align*}
    \frac{\partial  M^{(q)}(\gamma,p)}{\partial \gamma}=\tilde{H}_q(\zeta)+\gamma\zeta'\tilde{H}_q'(\zeta)-H_q(\beta(\zeta))+(1-\gamma)\zeta'\beta'(\zeta)H_q'(\beta(\zeta))\;,
\end{align*}
at $\gamma=0$, $\zeta_0=\zeta(0)$ satisfies $\beta(\zeta_0)=p$, so
\begin{align*}
    \frac{\partial  M^{(q)}(\gamma,p)}{\partial \gamma}\bigg\rvert_{\gamma=0}=\tilde{H}_q(\zeta_0)-H_q(p)+\zeta'(0)\beta'(\zeta_0)H_q'(p)\;.
\end{align*}
The derivative $\zeta'$ exists by the
implicit function theorem.

At the same time, applying the derivative with respect to $\gamma$ to the constraint, we get $\frac{\zeta}{2}+\gamma\frac{\zeta'}{2}-\beta(\zeta)+(1-\gamma)\zeta'\beta'(\zeta)=-\frac{1}{2}$, thus at $\gamma=0$ we have $\zeta'(0)\beta'(\zeta_0)=p-\frac{1+\zeta_0}{2}$, so
\begin{align}
\label{eq:Mq_prime_0}
    \frac{\partial  M^{(q)}(\gamma,p)}{\partial \gamma}\bigg\rvert_{\gamma=0}=\tilde{H}_q(\zeta_0)-H_q(p)+\left(p-\frac{1+\zeta_0}{2}\right)H_q'(p)\;.
\end{align}
 The fact that $\beta(\zeta_0)=p$ gives us that $\frac{1-\zeta_0}{\zeta_0}=\sqrt{\frac{1-p}{p C_q}}$ i.e., $\zeta_0=\frac{1}{1+\frac{2}{q-2}\sqrt{\frac{(q-1)(1-p)}{p}}}$. Recall also that $H_q'(p)=\log_q\left(\frac{(q-1)(1-p)}{p}\right)$. Let us rewrite $\tilde{H}_q(\zeta_0)$ as 
 \begin{align*}
     \tilde{H}_q(\zeta_0)&=-(1-\zeta_0)\log_q\left(\frac{1-\zeta_0}{2}\right)-\zeta_0\log_q\left(\frac{\zeta_0}{q-2}\right)
     \\&=-\log_q\left(\frac{1-\zeta_0}{2}\right)+\zeta_0\log_q\left(\frac{q-2}{2}\sqrt{\frac{1-p}{p C_q}}\right)
     \\&=-\log_q\left(\frac{\sqrt{\frac{(q-1)(1-p)}{p}}}{q-2+2\sqrt{\frac{(q-1)(1-p)}{p}}}\right)+\zeta_0\log_q\left(\sqrt{\frac{(q-1)(1-p)}{p}}\right)
     \\&=-\frac{(1-\zeta_0)}{2}H_q'(p)+\log_q\left(q-2+2\sqrt{\frac{(q-1)(1-p)}{p}}\right)\;.
 \end{align*}
 Substituting the above into \Cref{eq:Mq_prime_0}, we get
 \begin{align*}
     \frac{\partial  M^{(q)}(\gamma,p)}{\partial \gamma}\bigg\rvert_{\gamma=0}&=(p-1)H_q'(p)-H_q(p)+\log_q\left(q-2+2\sqrt{\frac{(q-1)(1-p)}{p}}\right)
     \\&=\log_q\left(\frac{p}{q-1}\right)+\log_q\left(q-2+2\sqrt{\frac{(q-1)(1-p)}{p}}\right)
     \\&=\log_q\left(\frac{q-2}{q-1}p+2\sqrt{\frac{p(1-p)}{q-1}}\right)\;.
     \pushQED{\qed}\qedhere\popQED
 \end{align*}

\section{BAWGN channel}
\label{sec:bawgn}

We follow the exposition in \cite[Section 4.1]{RU08}.
In the $\BAWGN_{\sigma^2}$ channel, the input alphabet
is $\mathbb{F}_2$, and the output distribution for $b\in\{0,1\}$
is given by the standard Gaussian distribution $\mathcal{N}((-1)^b, \sigma^2)$. In other words, conditioned on a transmitted codeword
$c=(c_1,\ldots,c_n)$, letting 
$\overline{c}=((-1)^{c_1},\ldots,(-1)^{c_n})$, the noisy output
can be written as $Y=\overline{c}+Z$, where $Z$ is a vector
of $n$ i.i.d.~$\mathcal{N}(0,\sigma^2)$ Gaussians. In particular,
the density function of $Z$ is 
$f_Z(x)=\frac{1}{\left(\sqrt{2\pi}\sigma\right)^{n}}
\exp\left(-\frac{\|x\|^2}{2\sigma^2}\right)$.
Note that for two codewords $x,y\in\mathbb{F}_2^n$ with
Hamming distance $k$, the Euclidean distance of their embeddings
is $\|\overline{x}-\overline{y}\|_2=2\sqrt{k}$. 

In the following, let $B(\overline{x},R)
\coloneqq\{y:\|y-x\|_2\le R\}$ denote the Euclidean ball
of radius $R$ around $\overline{x}\in\mathbb{R}^n$. 
Recall that the Lebesgue volume of such a ball is given
by $V_n(R)=\frac{\pi^{n/2}}{\Gamma(n/2+1)}R^n$.
Furthermore, let $I_n(R,D)$ denote the Lebesgue volume of the
intersection $B(x,R)\cap B(y,R)$ for $\overline{x},\overline{y}$ such that
$\|\overline{x}-\overline{y}\|=D$.

We can now state our main result,
in a similar fashion as
in \Cref{Thm:1}.

\begin{definition}
    For $0<\lambda\le 1$, $\sigma^2>0$ and $0\le\gamma<\sigma^2$, let 
    \begin{align*}
        F^{\BAWGN}_{\lambda}(\gamma,\sigma^2)\coloneqq
        -\frac{1}{2}\log\left(1-\frac{\gamma}{\sigma^2}\right)
        +\gamma\log(2^{\lambda}-1)\;.
    \end{align*}
    Then, for $0<\lambda<1$ and $0\le\delta\le 1$, let
    \begin{align*}
        \sigma^2_*(\lambda,\delta)\coloneqq
        \inf\left\{\sigma^2>\delta: F^{\BAWGN}_{\lambda}(\delta,\sigma^2)\le 0\right\}.
    \end{align*}
\end{definition}
The value $\sigma^2_*(\lambda,\delta)$ is well defined, since
$F^{\BAWGN}_{\lambda}(\delta,\sigma^2)\le 0$ for sufficiently
large $\sigma^2$.

\begin{theorem}
\label{thm:3}
    Let $0<\lambda,\delta<1$, let $\{\cC_n\}_{n\ge 1}$ be a family of binary linear codes with $\delta(\cC_n)\ge \delta$, and such that $\underset{n\to \infty}{\lim} P_b(\cC_n,\BEC_\lambda)=0$. Then,
\begin{align*}
    \underset{n\to \infty}{\lim} P_B\left(\cC_n,\BAWGN_{\sigma^2}\right)=0,\ \text{for every}\ \sigma^2<\sigma_*^2(\lambda,\delta)\;.
\end{align*}
\end{theorem}

In the case of symmetric channels, we were using the Poltyrev bound
in order to bound the block error probability.
For the binary Gaussian channel, in~\cite{poltyrev2002bounds} 
Poltyrev proposed what is known as the tangential sphere bound. For our purposes, it seems
sufficient to use a simpler sphere bound, see
\cite{HP94, sason2006performance}. This bound is combined
with an estimate of $I_n(\alpha\sqrt{n},\beta\sqrt{n})$,
for constant $\alpha,\beta>0$.

\begin{lemma}
\label{lem:sphere_bound}
Let $\mathcal{C}\subseteq\mathbb{F}_2^n$ be a linear code with
weight distribution
$(A_0,A_1,\ldots,A_n)$,
$0\le s\le 4$
and $\sigma^2>0$. Then,
\begin{align*}
P_{B}(\mathcal{C},\BAWGN_{\sigma^2})
&\le
\frac{\exp(-n/2+sn/2)}
{\left(\sqrt{2\pi}\sigma\right)^n}
\sum_{1\le w\le n}A_w
I_n\left((1+s)\sigma\sqrt{n},2\sqrt{w}\right)
+2\exp\left(-\frac{s^2n}{16}\right)\;.
\end{align*}
\end{lemma}

\begin{lemma}
\label{lem: Gaussian_ball_intersection}
Let $\zeta,\gamma\ge 0$.
If $\gamma\ge \zeta$, then
$I_n(\zeta\sqrt{n},2\gamma\sqrt{n})=0$. Otherwise,
\begin{align*}
\frac{1}{80n^2}\sqrt{\frac{\zeta-\gamma}{\zeta+\gamma}}
\left(2e\pi\right)^{n/2}(\zeta^2-\gamma^2)^{n/2}
\le I_n(\zeta\sqrt{n},2\gamma\sqrt{n})
\le\frac{\sqrt{2e}}{\pi}\sqrt{\frac{\zeta-\gamma}{\zeta+\gamma}}\left(2e\pi\right)^{n/2}(\zeta^2-\gamma^2)^{n/2}
\end{align*}
\end{lemma}
For completeness, \Cref{lem:sphere_bound} 
and \Cref{lem: Gaussian_ball_intersection} are proved
in \Cref{app:sphere_bound_proof}.

\begin{corollary}
\label{cor:Gaussian_error_prob}
    Let $\sigma^2> 0$, $0\le s\le 4$ and $t=(1+s)^2\sigma^2n$. Let us write $w=\gamma n$.
    For a binary linear code $\cC$ with $\delta(\cC)\ge\delta$ it holds
    \begin{align*}
        P_B(\cC, \BAWGN_{\sigma^2})\le \frac{\sqrt{2e}\exp(\frac{sn}{2})}{\pi}\sum_{\delta n\le w\le \min(n,t)}A_{w}\left((1+s)^2-\frac{\gamma}{\sigma^2}\right)^{\frac{n}{2}}+2\exp(-s^2n/16) \;.
    \end{align*}
\end{corollary}
\begin{proof}
    The proof is a straightforward implication of \Cref{lem:sphere_bound} and \Cref{lem: Gaussian_ball_intersection} with $\zeta=(1+s)\sigma$ and the assumption $\delta(\cC)\ge\delta$.
\end{proof}

\begin{lemma}
    Let $0 <\lambda\le 1$.
    \begin{enumerate}[1)]
        \item For $\sigma>0$, the function $\gamma\mapsto F^{\BAWGN}_{\lambda}(\gamma,\sigma^2)$ is strictly convex for $0\le \gamma<\sigma^2$.
        \item For $\gamma>0$, the function $\sigma\mapsto F^{\BAWGN}_{\lambda}(\gamma,\sigma^2)$ 
        is strictly decreasing for $\sigma^2>\gamma$.
    \end{enumerate}
\end{lemma}
\begin{proof}
    \begin{enumerate}[1)]
        \item It is easy to see that for $0\le\gamma<\sigma^2$ \begin{align*}
            \frac{\partial^2 F^{\BAWGN}_{\lambda}(\gamma,\sigma^2)}{\partial\gamma^2}=\frac{1}{2(\ln 2)(\sigma^2-\gamma)^2}>0\;. 
        \end{align*}
        \item This is clear,
        since $1-\frac{\gamma}{\sigma^2}$ is strictly
        increasing in $\sigma^2$ for
        $\sigma^2>\gamma$.\qed
    \end{enumerate}
    \renewcommand{\qedsymbol}{}
\end{proof}

\begin{proof}[Proof of Theorem~\ref{thm:3}] Applying \Cref{cor:Gaussian_error_prob} with $s\coloneqq n^{-1/4}$
and writing $w=\gamma n$,
\begin{align*}
        P_B(\cC_n, \BAWGN_{\sigma^2})&\le \frac{\sqrt{2e}\exp(\frac{n^{3/4}}{2})}{\pi}\sum_{w\in D}A_{\gamma n}\left((1+n^{-1/4})^2-\frac{\gamma}{\sigma^2}\right)^{\frac{n}{2}}+2\exp(-n^{1/2}/16) 
        \\&\le \exp(o(n))\cdot
        \underset{w\in D}{\sup} A_{\gamma n} \left((1+n^{-1/4})^2-\frac{\gamma}{\sigma^2}\right)^{\frac{n}{2}}+o(1)\;,
\end{align*}
where $D\coloneqq
\{w:\delta n\le w\le n,
w<(1+n^{-1/4})^2\sigma^2)\}$.
(Note that for $w=(1+n^{-1/4})\sigma^2$,
the relevant term in the sum is
zero, so we can omit it.)
If $\sigma^2<\delta$, then $D=\emptyset$ for large enough $n$
and we are already done. Therefore, assume $\delta\le\sigma^2$,
which implies $D\ne\emptyset$ for large $n$.
Then,
\begin{align*}
   A_{\gamma n} \left((1+n^{-1/4})^2-\frac{\gamma}{\sigma^2}\right)^{\frac{n}{2}}=\exp_2\left[n\left(\frac{\log(A_{\gamma n})}{n}+\frac{1}{2}\log\left(1-\frac{\gamma}{\sigma^2}+2n^{-1/4}+n^{-1/2}\right)\right)\right].
\end{align*}
Since $\underset{n\to\infty}{\lim}P_b(\cC_n,\BEC_{\lambda})=0$, 
it follows $H(X|Y)=o(n)$  
for $X$ uniform on $\cC_n$
and $Y$ the output of $X$
after transmitting through the BEC
(see~\Cref{cl:pbit-implies-entropy}).
By \Cref{prop:weight_distribution_bound},
   $\frac{\log(A_{\gamma n})}{n}\le-\gamma\log(2^{\lambda}-1)+o(1)$.
   Therefore, 
   \begin{align*}
       P_B(\cC_n, \BAWGN_{\sigma^2})&\le \exp\left(o(n)-n
       \inf_{\substack{\delta\le\gamma\le 1\\ \gamma<(1+s)^2\sigma^2}}\left(
       -\frac{1}{2}\log\left(1-\gamma/\sigma^2+3n^{-1/4}\right)
       +\gamma \log(2^{\lambda}-1)\right)\right)+o(1)\;.
   \end{align*}
Let $\theta\coloneqq 2^\lambda-1$ and $R_n(\gamma)\coloneqq -\frac{1}{2}\log\left(1-\gamma/\sigma^2+3n^{-1/4}\right)$. It remains to show that
$R_n(\gamma)+\gamma\log\theta$ is positive and bounded away from zero, uniformly in $\gamma$. 

To that end, consider two cases.
If $R_n(\gamma)\ge -2\log \theta$, then
\begin{align*}
    R_n(\gamma)+\gamma\log\theta \ge-(2-\gamma)\log\theta\ge-\log\theta>0\;.
\end{align*}
On the other hand, if $R_n(\gamma)< -2\log \theta$, then, $1-\gamma/\sigma^2>\theta^4-3n^{-1/4}$, thus for $n$ large enough, we have 
$1-\gamma/\sigma^2>\frac{\theta^4}{2}$. Also, since the logarithmic function is concave, we have 
\begin{align*}
    R_n(\gamma)&\ge -\frac{1}{2}\log\left(1-\gamma/\sigma^2\right)-\frac{1}{(2\ln2)(1-\frac{\gamma}{\sigma^2})}\cdot3n^{-1/4}\\&\ge -\frac{1}{2}\log\left(1-\gamma/\sigma^2\right)-\frac{3n^{-1/4}}{(\ln 2)\theta^4}.
\end{align*}
Note that in this case $1-\gamma/\sigma^2>\frac{\theta^4}{2}$
implies $\delta\le\gamma<\left(1-\frac{\theta^4}{2}\right)\sigma^2
\le\sigma^2$. In particular,
$F_\lambda^{\BAWGN}(\delta,\sigma^2)$ and
$F_\lambda^{\BAWGN}(\gamma,\sigma^2)$ are well defined.
Then,
\begin{align*}
     R_n(\gamma)+\gamma\log\theta&\ge F^{\BAWGN}_{\lambda}(\gamma,\sigma^2)+o(1)
     \ge F^{\BAWGN}_\lambda(\delta,\sigma^2)+o(1)>0\;, 
\end{align*}
which is also uniformly bounded
away from 0 for large $n$.
$F^{\BAWGN}_{\lambda}(\delta,\sigma^2)\le
F^{\BAWGN}_\lambda(\gamma,\sigma^2)$ follows
since 
$F_\lambda^{\BAWGN}(0,\sigma^2)=0$
and due to the convexity of $\gamma\mapsto F_\lambda^{\BAWGN}(\gamma,\sigma^2)$.
$F^{\BAWGN}_{\lambda}(\delta,\sigma^2)>0$ is implied by the
assumption $\sigma^2<\sigma_*^2(\lambda,\delta)$.
All in all, for $n$ large enough, we have
\begin{align*}
       P_B(\cC_n, \BAWGN_{\sigma^2})&\le o(1)\;,
   \end{align*}
which concludes the proof.
\end{proof}

To conclude, let us show that taking $\sigma_*^2$ in the limit $\delta\to 0$ 
results in the bound given in~\cite{hkazla2021codes} (see \Cref{thm:original}).

\begin{lemma}
      Let $0<\lambda<1$. Then the function $\delta\mapsto\sigma^2_*(\lambda,\delta)$ is 
      strictly increasing for $0<\delta\le 1$. Furthermore,
    \begin{align*}
        \underset{\delta\to 0^+}{\lim}\sigma^2_*(\lambda,\delta)=-\frac{1}{2\ln(2^\lambda-1)}\;.
    \end{align*}
\end{lemma}
\begin{proof}
Fix $0<\lambda<1$
and let us write $F(\gamma,\delta)\coloneqq 
F_\lambda^{\BAWGN}(\gamma,\delta)$.
First, let us show that $\sigma^2_*(\lambda,\delta)$
is increasing in $\delta$. 
Observe that
$\lim_{\sigma^2\to \delta^+}
F(\delta,\sigma^2)=\infty$
for $\delta>0$, which implies
$\sigma_*^2(\lambda,\delta)>\delta$.
Let $0<\delta_0<\delta_1\le 1$ and
$\sigma_0^2\coloneqq \sigma_*^2(\lambda,\delta_0)$,
$\sigma_1^2\coloneqq \sigma_*^2(\lambda,\delta_1)$.
Since $\sigma_1^2>\delta_1>\delta_0$, the values
of $F(\delta_0,\sigma_1^2)$
and $F(\delta_1,\sigma_1^2)$ are well defined.

Now we can proceed like in \Cref{lem:delta_limit}. 
If $F(\delta_0,\sigma_1^2)<0$, then by the definition
of $\sigma_0^2$ it follows $\sigma_0^2\le\sigma_1^2$.
In fact, by the continuity of $F$ it must be $\sigma_0^2<\sigma_1^2$. On the other hand,
if $F(\delta_0,\sigma_1^2)\ge 0$, then,
by $F(0,\sigma_1^2)=0$ and the strict convexity of $F$
we have $F(\delta_1,\sigma_1^2)>0$, which is a contradiction
by the definition of $\sigma_1^2$ and continuity of $F$.
So, $\sigma_0^2<\sigma_1^2$ and $\sigma_*^2$ is strictly
increasing in $\delta$.

Let us turn to the limit as $\delta\to 0^+$. Note that
 \begin{align*}
     \frac{\partial F(\gamma,\sigma^2)}{\partial\gamma}\bigg\rvert_{\gamma=0}=\log(2^{\lambda}-1)+\frac{1}{2(\ln 2)\sigma^2}\;,
 \end{align*}
and the rest of the proof follows like in \Cref{lem:delta_limit} 
using the critical value $\sigma^2_0:=-\frac{1}{2\ln(2^{\lambda}-1)}$.
\end{proof}

\section*{Acknowledgements}
We thank Venkatesan Guruswami for a helpful discussion.
The authors were partially
supported by the Alexander von Humboldt Foundation German research chair
funding and the associated DAAD projects No. 57610033
and 57761435.

\appendix
\section{Appendix}
\subsection{Effective version
of the bound from~\cite{rudra2010two}}
\label{app:effective-ru}

In the proof of Theorem~1, part~(a)
in~\cite{rudra2010two}, one can
see that for a linear code with $\delta(\cC)\ge \delta$
and weight of an error pattern $\rho\le \delta-\epsilon$,
the proportion of ``bad'' error patterns which might result in the failure
of MAP block decoding is bounded by
\begin{align*}
n\cdot \text{(\# of choices for $A$)}\cdot (q-1)^{-\eps n+1}\;,
\end{align*}
where $A$ is a certain set of size at least $(1-\delta+\eps) n$.
Therefore, the number of choices for $A$ can be upper bounded
by $\sum_{k\le (\delta-\eps)n}\binom{n}{k}\le \exp_2\left(nh(\min(1/2, \delta-\eps)\right)$. Accordingly, a code family of relative
distance $\delta$ has vanishing error probability whenever
\begin{align}
\label{eq:21}
h(\min(1/2,\delta-\eps))-\eps\log(q-1)<0\;.
\end{align}
As the left hand side of~\eqref{eq:21} is strictly decreasing
in $\eps$, one can compute the unique value $\eps_0(q,\delta)$
such that there is vanishing error probability for $p<p_0(q,\delta)=\delta-\eps_0(q,\delta)$.

We observed that
$p_0(q,\delta)<J_q(\delta)$ for $q\le 19$
and $0<\delta\le 1-1/q$ by plotting and visual inspection.
For an illustration, see \Cref{fig:effective-ru}
with a plot for $q=15$. 

It seems that this effective version is tight (in certain regime of $\delta$) for large $q$, as it matches the TVZ upper bound~\cite{tsfasman1982modular,ihara1981some} for moderate values of $\delta$. However, it becomes loose for large $\delta$. We illustrate it for $q=2^{20}$ in \Cref{fig:all bounds}. 

\begin{figure}[!ht]
    \centering
    \includegraphics[width=0.7\textwidth]{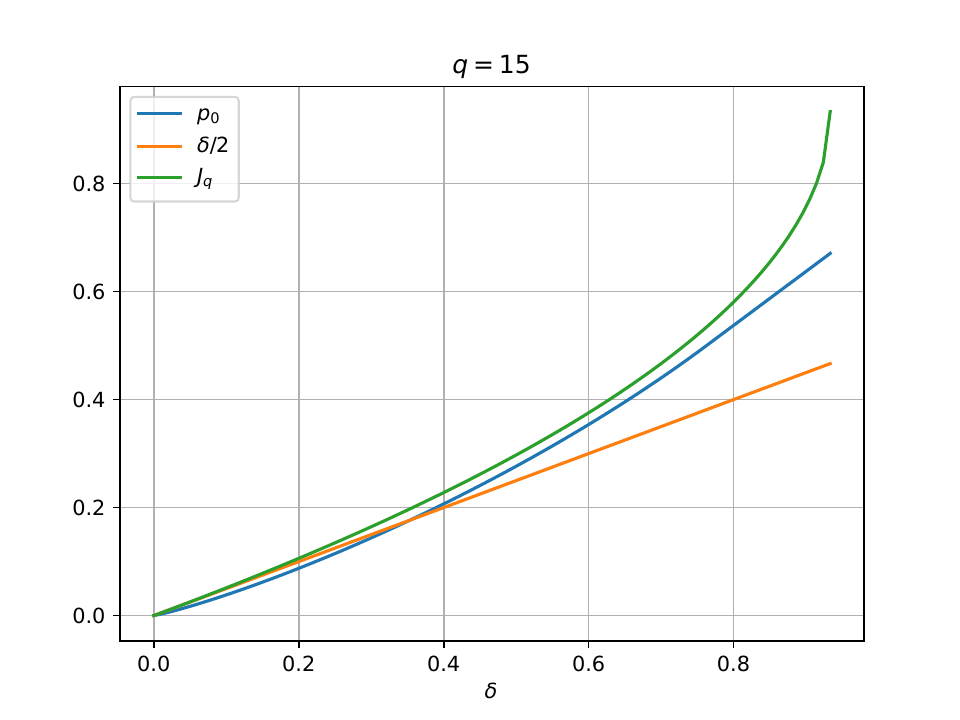}
    \caption{$p_0(15,\delta)$ as a function
    of $\delta$ and compared to $J_q(\delta)$
    and $\delta/2$. See \Cref{app:effective-ru}.}
    \label{fig:effective-ru}
\end{figure}

\begin{figure}[!ht]
    \centering
    \includegraphics[width=0.7\linewidth]{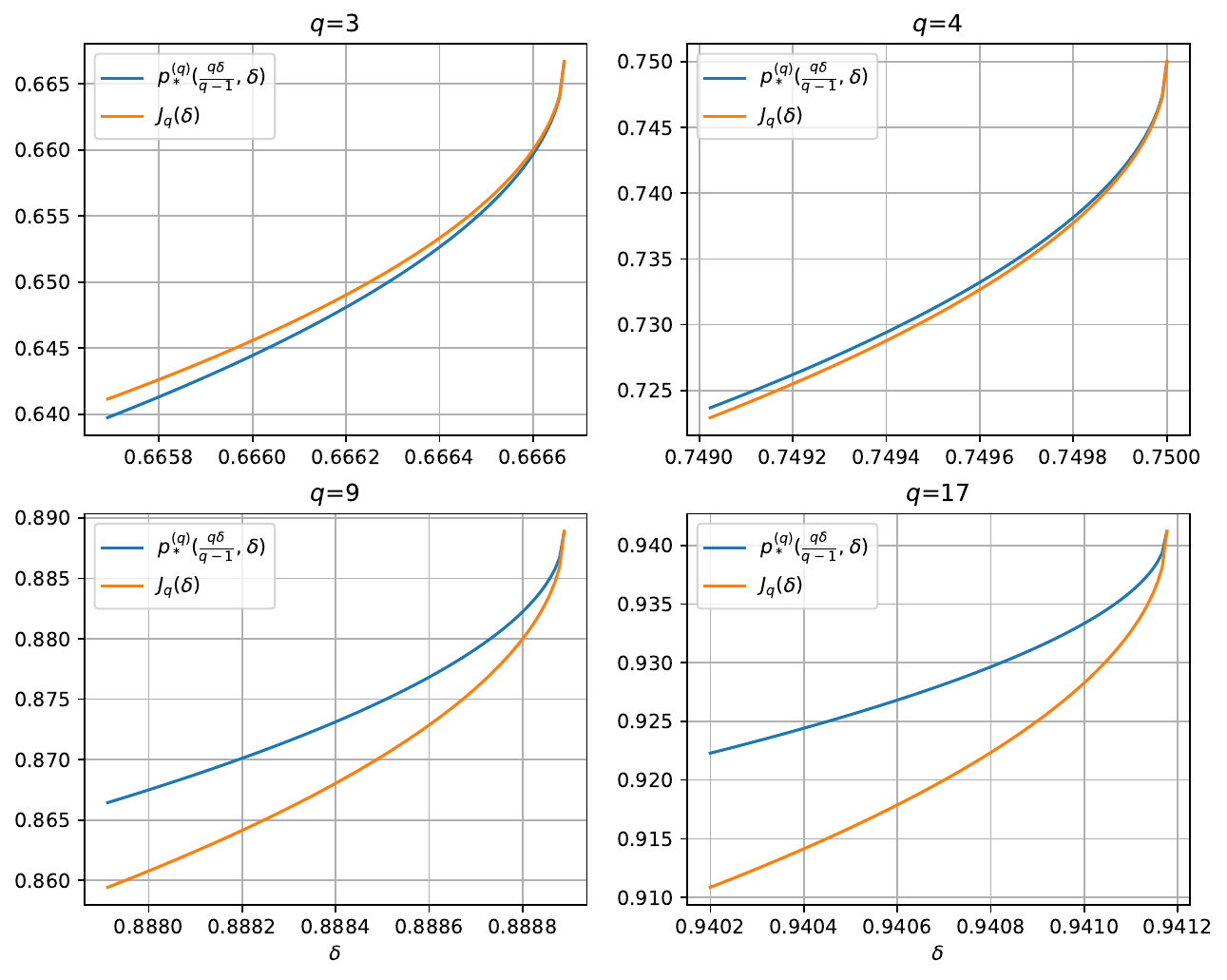}
    \caption{Our lower bound $p_*^{(q)}(\frac{q\delta}{q-1},\delta)$ from \Cref{thm:unconditional_bound} 
    plotted against the Johnson bound for
    $q\in\{3,4,9,17\}$ for $\delta\in[1-\frac{1}{q}-2^{-10}, 1-\frac{1}{q}]$. For $q\ge4$, our bound is larger than $J_q(\delta)$ for larger values of $\delta$.}
    \label{fig:pstar_vs_jq_large_delta}
\end{figure}

\begin{figure}[!ht]
    \centering
    \includegraphics[width=\linewidth]{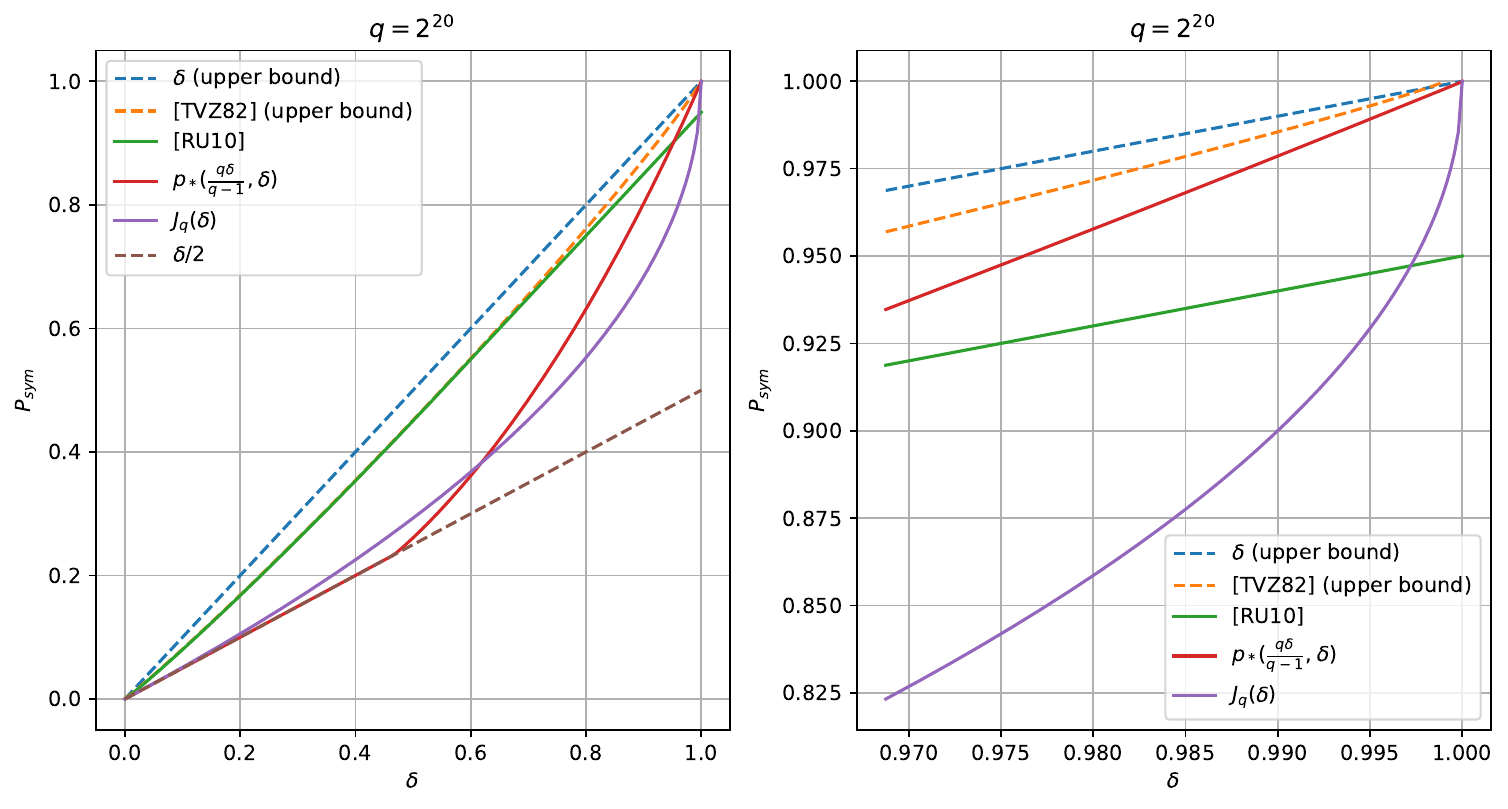}
    \caption{Upper and lower bounds for $\psym(q,\delta)$ discussed in the paper, for large $q=2^{20}$. Note that the bound $p^{(q)}_*(\frac{\delta q}{q-1},\delta)$ is only for linear codes.}
    \label{fig:all bounds}
\end{figure}

\subsection{Proof of \Cref{prop:poltyrev}}
\label{apx: A1}
    Assume w.l.o.g.~that $\Sigma=\mathbb{Z}_q$, so
    that we can use arithmetic in $\mathbb{Z}_q$.    
    Consider the MAP decoder for $\cC$ which,
    given $y$, outputs
    a codeword which is closest to $y$
    in the Hamming distance, breaking ties in an arbitrary way. 
    Let $t'\coloneqq p n-n^{\theta}$. 
    Consider the error vector $Z$ distributed according
    to $\Pr[Z=z]=(\frac{p}{q-1})^{\wt(z)}(1-p)^{n-\wt(z)}$.
    Assume that $x\in\cC$ is the original message,
    and note that $\Delta(x+Z,x)=\wt(Z)$.    
    By the union bound,
    \begin{align*}
    P_{B}(\cC,\qSC_p,x)
    &\le \Pr[\wt(Z)\notin (t',t)]
    +\Pr[\exists y\in\cC:1\le\Delta(x,y)\le w_0\text{ and }\Delta(x+Z,y)\le\wt(Z)]\\
    &\qquad+\Pr[\exists y\in\cC:\Delta(x,y)>w_0, \Delta(x+Z,y)\le t\text{ and }
    t'<\wt(Z)<t]\\
    &=:A_1+A_2+A_3\;,
    \end{align*}
    where $A_1,A_2,A_3$ denote the respective terms on the 
    right hand side. We will bound each of $A_1,A_2,A_3$
    with respective terms in~\eqref{eq:22} and~\eqref{eq:23}.
    First, by the Hoeffding bound,
    $A_1\le 2\exp\left(-2n^{2\theta-1}\right)$.

    For $A_2$ we will use the union bound.
    Consider $y\in\cC$ with $\Delta(x,y)=w$. By a standard
    argument, it holds
    \begin{align*}
        \Pr[\Delta(x+Z,y)\le\wt(Z)]\le Z(\qSC_p)^w\;.        
    \end{align*}
    Indeed, for $1\le i\le n$ such that $x_i\ne y_i$ define a random variable
    \begin{align*}
    \alpha_i&\coloneqq
    \begin{cases}
    \sqrt{\frac{p}{(q-1)(1-p)}}&\text{if $Z_i=0$,}\\
    \sqrt{\frac{(q-1)(1-p)}{p}}&\text{if $Z_i=y_i-x_i$,}\\
    1&\text{otherwise.}
    \end{cases}
    \end{align*}
    Since $\sqrt{\frac{p}{(q-1)(1-p)}}\le 1$, the event
    $\Delta(x+Z,y)\le \wt(Z)$ implies that
    $\prod_{i:x_i\ne y_i}\alpha_i\ge 1$. Furthermore, it 
    holds $\E\alpha_i=Z(\qSC_p)$. Hence,
    \begin{align*}
    \Pr[\Delta(x+Z,y)\le\wt(Z)]
    &\le \Pr\left[\prod_{i:x_i\ne y_i}\alpha_i\ge 1\right]
    \le \E\left[\prod_{i:x_i\ne y_i}\alpha_i\right]
    =Z(\qSC_p)^w\;.
    \end{align*}
    Applying the union bound, it follows
    \begin{align*}
        A_2&\le
        \sum_{w=1}^{w_0} A_w(x)
        Z(\qSC_p)^w\;.
    \end{align*}

It remains to estimate $A_3$. Let $F_t(z)$
be the number of codewords $y\in\cC$ with $\Delta(y,x)>w_0$
and $\Delta(x+z,y)\le t$.
Using $p\le 1-1/q$,
\begin{align*}
    A_3&=\sum_{z:t'<\wt(z)<t}\mathds{1}[F_{t}(z) \ge 1]
    (1-p)^{n-\wt(z)}\left(\frac{p}{q-1}\right)^{\wt(z)}
    \\&\le
    \left(\frac{(1-p)(q-1)}{p}\right)^{n^{\theta}}
    q^{-nH_q(p)}
    \sum_{z:t'<\wt(z)<t}\mathds{1}[F_{t}(z) \ge 1]\;.
\end{align*}
Let $S\coloneqq \sum_{z:t'<\wt(z)<t}\mathds{1}[F_t(z)\ge 1]$.
Comparing with~\eqref{eq:23}, we will be done if we show
$S\le\sum_{w>w_0}A_w(x)\mu_q(n,t,w)$.
Indeed,
\begin{align*}
    S&\le\sum_{z\in B(0,t)}\mathds{1}[F_t(z)\ge 1]
    \le\sum_{z\in B(0,t)}    
    F_t(z)
    \\&=
    \sum_{z\in B(0,t)}\sum_{y\in\cC:\Delta(x,y)>w_0}\mathds{1}[
    y\in B(x+z,t)]\\
    &=\sum_{y\in\cC:\Delta(x,y)>w_0}
    \left|B(0,t)\cap B(y-x,t)\right|
    =\sum_{w=w_0+1}^{n}A_w(x)\mu_q(n,t,w)\;.
    \pushQED{\qed}\qedhere\popQED
\end{align*}

\subsection{Proofs of \Cref{lem:sphere_bound}
and \Cref{lem: Gaussian_ball_intersection}}
\label{app:sphere_bound_proof}
As it requires only a small additional effort,
we extend \Cref{lem:sphere_bound}
to a slightly more general statement about list decoding.
Below, $P_{B,L}$ denotes the MAP block error
probability for list decoding for list of size $L$.

\begin{lemma}
Let $\mathcal{C}\subseteq\mathbb{F}_2^n$ be a linear code with
weight distribution
$(A_0,A_1,\ldots,A_n)$,
$L\ge 1$ an integer,
$0\le s\le 4$
and $\sigma^2>0$. Then,
\begin{align*}
P_{B,L}(\mathcal{C},\BAWGN_{\sigma^2})
&\le
\frac{1}{L}
\frac{\exp(-n/2+sn/2)}
{\left(\sqrt{2\pi}\sigma\right)^n}
\sum_{1\le w\le n}A_w
I_n\left((1+s)\sigma\sqrt{n},2\sqrt{w}\right)
+2\exp\left(-\frac{s^2n}{16}\right)\;.
\end{align*}
\end{lemma}

\begin{proof}
Let $\mathcal{R}(s)$ denote the event
$(1-s)\sigma^2 n\le \|Z\|^2\le (1+s)\sigma^2n$,
and note that this event implies
$(1-s)\sigma\sqrt{n}\le\|Z\|\le(1+s)\sigma\sqrt{n}$.
Note that the distribution of 
$\frac{1}{\sigma^2}\|Z\|^2$ is chi-squared with
$n$ degrees of freedom. By standard tail bounds, the norm $\|Z\|$ must be close to $\sigma\sqrt{n}$
with high probability. For example, applying~\cite[Lemma~1]{laurent2000adaptive}, it follows that, for $0\le s\le 4$, it holds
$\Pr[\mathcal{R}(s)^c]=\Pr\left[\left|\frac{1}{\sigma^2}\|Z\|^2-n\right|>
s n\right]\le 2\exp\left(-\frac{s^2n}{16}\right)$.

Let us enumerate the code $\mathcal{C}=\{0^n,c_1,\ldots,c_M\}$.
Since $\mathcal{C}$ is linear, let us
assume w.l.o.g.~that the embedding of the zero codeword
$\overline{0^n}=1^n$ is selected for transmission.
Consider a decoder which maps 
the received noisy message $Y=1^n+Z$ to the list
of $L$ codewords with embeddings which are closest in the Euclidean distance.
The true codeword will be in that list unless
there exist at least $L$ other codewords with
embeddings closer to $Y$ than $1^n$. If we condition on the event
$\mathcal{R}(s)$, this can happen only if there exist at
least $L$ embeddings of nonzero codewords
such that their distance
to $Y$ is at most $(1+s)\sigma\sqrt{n}$. Accordingly,
let $W=\left|\left\{i>0: \|\overline{c_i}-Y\|\le(1+s)\sigma\sqrt{n}\right\}\right|$. From those considerations, we have
\begin{align*}
P_{B,L}(\mathcal{C},\BAWGN_{\sigma^2})
\le \Pr[W\ge L\text{ and }\mathcal{R}(s)] + \Pr[\mathcal{R}(s)^c]\;,
\end{align*}
and we already established
$\Pr[\mathcal{R}(s)^c]\le
2\exp(-s^2n/16)$. It remains
to analyze the first term:
\begin{align*}
\Pr[W\ge L\text{ and }\mathcal{R}(s)]
&\le \frac{1}{L}\E\left[W \cdot \mathds{1}(\mathcal{R}(s))\right]
=\frac{1}{L}\sum_{i=1}^M \Pr\left[Z\in B\left(\overline{c_i}-1^n,(1+s)\sigma\sqrt{n}\right)\cap \mathcal{R}(s)\right]\\
&=\frac{1}{L}\sum_{i=1}^M\int_{\mathbb{R}^n}
\mathds{1}\left(z\in B\left(\overline{c_i}-1^n,(1+s)\sigma\sqrt{n}\right)\right)\cdot \mathds{1}\left(z\in\mathcal{R}(s)\right)
f_Z(z)\,\mathrm{d}z\\
&\le\frac{1}{L}
\frac{\exp\left(-n/2+sn/2\right)}{\left(\sqrt{2\pi}\sigma\right)^{n}}
\sum_{i=1}^M\int_{\mathbb{R}^n}
\mathds{1}\left(z\in B\left(\overline{c_i}-1^n,(1+s)\sigma\sqrt{n}\right)\right)\cdot 
\mathds{1}\left(z\in\mathcal{R}(s)\right)\,\mathrm{d}z
\\&\le\frac{1}{L} \frac{\exp(-n/2+sn/2)}{\left(\sqrt{2\pi}\sigma\right)^{n}}
\sum_{1\le w\le n}A_w I_n\left(
(1+s)\sigma\sqrt{n},2\sqrt{w}
\right)\;.\qedhere
\end{align*}
\end{proof}

\begin{proof}[Proof of \Cref{lem: Gaussian_ball_intersection}.]
If $\gamma\ge \zeta$, the statement
is clear. Otherwise,
for notational convenience let us
consider the $(n+1)$-dimensional volume $I_{n+1}(R,D)$
for $R\coloneqq\zeta\sqrt{n+1}$
and $D\coloneqq 2\gamma\sqrt{n+1}$.

Consider the line connecting the centers of the two balls. We compute the volume as the integral
of the $n$-dimensional volume along this line. Starting
from the center of the first ball, the balls intersect at the range of
distances $D-R\le r\le R$. The intersection consists of two symmetric
spherical caps. Let us consider the cap for $D/2\le r\le R$. For every
$r$, the intersection is the $n$ dimensional ball of radius $\sqrt{R^2-r^2}$. Accordingly,
\begin{align*}
I_{n+1}(R,D)&=
\frac{2\pi^{n/2}}{\Gamma\left(\frac{n}{2}+1\right)}\int_{D/2}^{R}
\left(\sqrt{R^2-r^2}\right)^{n}\,\mathrm{d}r
=\frac{2\pi^{n/2}R^n}{\Gamma\left(\frac{n}{2}+1\right)}\int_{D/2}^{R}
\left(\sqrt{1-(r/R)^2}\right)^{n}\,\mathrm{d}r\\
&=
\frac{2\pi^{n/2}R^{n+1}}{\Gamma\left(\frac{n}{2}+1\right)}\int_{D/(2R)}^{1}
\left(\sqrt{1-\beta^2}\right)^{n}\,\mathrm{d}\beta
=\frac{2\pi^{n/2}(n+1)^{(n+1)/2}}{\Gamma\left(\frac{n}{2}+1\right)}\zeta^{n+1}
\int_{\gamma/\zeta}^{1}
\left(\sqrt{1-\beta^2}\right)^{n}\,\mathrm{d}\beta\\
&=\frac{2\pi^{n/2}(n+1)^{(n+1)/2}}{\Gamma\left(\frac{n}{2}+1\right)}\int_{\gamma}^{\zeta}
\left(\sqrt{\zeta^2-\beta^2}\right)^{n}\,\mathrm{d}\beta\;.
\end{align*}
To upper bound the integral, we simply note that
$\int_{\gamma}^\zeta(\zeta^2-\beta^2)^{n/2}\,\mathrm{d}\beta
\le(\zeta^2-\gamma^2)^{n/2}(\zeta-\gamma)=\sqrt{\frac{\zeta-\gamma}{\zeta+\gamma}}(\zeta^2-\gamma^2)^{(n+1)/2}$. For the lower bound, note that
for any $0\le\eps\le 1/2$,
\begin{align*}
\int_{\gamma}^{\zeta}
(\zeta^2-\beta^2)^{n/2}\,\mathrm{d}\beta
\ge \int_{\gamma}^{\gamma+\eps(\zeta-\gamma)}
(\zeta-\gamma)^{n/2}(1-\eps)^{n/2}(\zeta+\gamma)^{n/2}\,\mathrm{d}\beta
\ge \eps\exp(-\eps n)(\zeta-\gamma)(\zeta^2-\gamma^2)^{n/2}\;.
\end{align*}
Taking $\eps=1/(2n^2)$, we obtain
$\int_{\gamma}^\zeta(\zeta^2-\beta^2)^{n/2}\,\mathrm{d}\beta
\ge \frac{1}{4n^2}(\zeta-\gamma)(\zeta^2-\gamma^2)^{n/2}$.

For the factor before the integral, we use the bound
$\sqrt{\frac{2\pi}{n}}\left(\frac{n}{e}\right)^n\le
\Gamma(n)\le 2\sqrt{\frac{2\pi}{n}}\left(\frac{n}{e}\right)^n$,
which is valid for all real $n\ge 1$. Accordingly,
with a little calculation it follows
\begin{align*}
\frac{1}{\sqrt{2}e^{3/2}\pi}
(2e\pi)^{(n+1)/2}
\le
\frac{2\pi^{n/2}(n+1)^{(n+1)/2}}{\Gamma\left(\frac{n}{2}+1\right)}
\le
\frac{\sqrt{2e}}{\pi}
(2e\pi)^{(n+1)/2}\;.
\end{align*}
Putting the bounds together, and substituting
$n+1\mapsto n$, the result follows.
\end{proof}

\subsection{Other deferred proofs}

\begin{claim}
    \label{cl:pbit-implies-entropy}
    Let $\{\cC_n\}_n$ be a family of linear codes over $\bF_q$, and let $X$ and $Y$ be respectively the input and output of transmitting 
    a uniformly chosen codeword from
    $\cC_n$ over $\qEC_\lambda$. If $P_{b}(\cC_n,\qEC_\lambda)=o(1)$, then $H(X|Y)=o(n)$.
\end{claim}
\begin{proof}
    On the erasure channel, it always holds $H(X_i|Y=y)\in\{0,\log q\}$. 
    If $H(X_i|Y=y)=0$, then $\Pr[\hat{x_i}^{\MAP}(Y)\ne X_i\;\vert\;Y=y]=0$.
    If $H(X_i|Y=y)=\log q$, then
    $\Pr[\hat{x_i}^{\MAP}(Y)\ne X_i\;\vert\;Y=y]=\frac{q-1}{q}$.
    Therefore, 
    $H(X_i|Y)=(\log q)\frac{q}{q-1} P_{b,i}(\cC_n, \qEC_\lambda)$.
    Thus,
\begin{align*}
    H(X|Y)&\le\sum_{i=1}^{n}H(X_i|Y)= \log(q)\cdot\frac{q}{q-1}\sum_{i=1}^{n}P_{b,i}(\cC_n,\qEC_\lambda)\le O\left(nP_{b}(\cC_n,\qEC_\lambda)\right)=o(n)\;.\qedhere
\end{align*}
\end{proof}

\begin{claim}
\label{prop:negative_G}
    For $0\le\lambda\le R< 1$ and $0\le\delta\le 1$, it holds
    $G^\perp_{\lambda,R}(\delta)\ge 0$.
\end{claim}
\begin{proof}
    Let $m_0\coloneqq 1-2^{\lambda-1}$ and $m\coloneqq \min(\delta,1-\delta)$.
    If $m<m_0$, then $G^\perp_{\lambda,R}(\delta)=R-\lambda-m\log(2^{1-\lambda}-1)\ge R-\lambda\ge 0$.
    On the other hand, observe that
    \begin{align*}
    G^\perp_{\lambda,R}(m_0)=h(m_0)-(1-R)=R-\lambda-m_0\log(2^{1-\lambda}-1)\ge 0\;.
    \end{align*}
    Therefore, if $m\ge m_0$, then 
    $G^\perp_{\lambda,R}(\delta)=h(m)-(1-R)\ge h(m_0)-(1-R)\ge 0$.
\end{proof}

\bibliography{references_arxiv_v1}
\bibliographystyle{alpha}
\end{document}

%% file: setting_arxiv_v1.tex
\usepackage[OT4]{fontenc}
\usepackage[margin=1.in]{geometry}
\usepackage{amssymb,amsfonts,amsmath,amsthm,amscd,dsfont,mathrsfs,pifont}
\usepackage{blkarray}
\usepackage{graphicx,float,psfrag,epsfig,color}
\usepackage{qtree}
\usepackage{comment}
\usepackage[dvipsnames]{xcolor}
\usepackage{authblk,textcomp,xcolor}
\usepackage{caption} 
\usepackage{subcaption}
\usepackage{mathtools}
\usepackage{enumerate}

\newcommand{\Htilde}{\tilde{H}}

\newcommand\independent{\protect\mathpalette{\protect\independent}{\perp}}
\def\independent#1#2{\mathrel{\rlap{$#1#2$}\mkern2mu{#1#2}}}

\definecolor{DSgray}{cmyk}{0,0,0,0.7}
\definecolor{DSred}{cmyk}{0,0.7,0,0.7}

\DeclareMathOperator*{\E}{\mathbb{E}}

\newcommand{\eps}{\varepsilon}

\DeclareMathOperator*{\argmax}{argmax}

\DeclareMathOperator{\qSC}{qSC}
\DeclareMathOperator{\qEC}{qEC}
\DeclareMathOperator{\BEC}{BEC}
\DeclareMathOperator{\wt}{wt}
\DeclareMathOperator{\BSC}{BSC}

\DeclareMathOperator{\MAP}{MAP}
\DeclareMathOperator{\BAWGN}{BAWGN}

\DeclareMathOperator{\poly}{poly}

\newcommand{\cE}{\mathcal{E}}

\newcommand{\cC}{\mathcal{C}}

\newcommand{\bF}{\mathbb{F}}
\newcommand{\psym}{p_{\mathrm{sym}}}
\newcommand{\pLsym}{p_{\mathrm{Lsym}}}

\usepackage{dsfont}

\usepackage[pdftex,pagebackref=true,colorlinks]{hyperref}
\hypersetup{linkcolor=[rgb]{.7,0,0}}
\hypersetup{citecolor=[rgb]{0,.7,0}}
\hypersetup{urlcolor=[rgb]{.7,0,.7}}
\hypersetup{pdftitle={Error Probalility Bounds on BSC-BEC},%
            colorlinks, linktocpage=true, pdfstartpage=1, pdfstartview=FitV,%
    breaklinks=true, pdfpagemode=UseNone, pageanchor=true, pdfpagemode=UseOutlines,%
    plainpages=false, bookmarksnumbered, bookmarksopen=true, bookmarksopenlevel=1,%
    hypertexnames=true, pdfhighlight=/O,%
    urlcolor=orange, citecolor=RoyalBlue}

\usepackage[capitalize,noabbrev]{cleveref}

\theoremstyle{plain}
\newtheorem{definition}{Definition}
\theoremstyle{plain}
\newtheorem{theorem}[definition]{Theorem}
\theoremstyle{plain}
\theoremstyle{plain}
\theoremstyle{plain}
\newtheorem{proposition}[definition]{Proposition}
\theoremstyle{plain}
\newtheorem{claim}[definition]{Claim}
\theoremstyle{plain}
\newtheorem{lemma}[definition]{Lemma}
\theoremstyle{plain}
\newtheorem{corollary}[definition]{Corollary}
\theoremstyle{plain}
\theoremstyle{plain}
\theoremstyle{remark}
\theoremstyle{remark}
\theoremstyle{remark}
\theoremstyle{plain}

%% file: main_arxiv_v1.bbl
\newcommand{\etalchar}[1]{$^{#1}$}
\begin{thebibliography}{KKM{\etalchar{+}}16}

\bibitem[AHO25]{abawonse2025generalized}
Olakunle Abawonse, Jan Hązła, and Ryan O'Donnell.
\newblock Generalized {S}amorodnitsky noisy function inequalities, with applications to error-correcting codes.
\newblock arXiv:2508.06940, 2025.

\bibitem[AS23a]{abbe2023proof}
Emmanuel Abbe and Colin Sandon.
\newblock A proof that {R}eed--{M}uller codes achieve {S}hannon capacity on symmetric channels.
\newblock In {\em Symposium on Foundations of Computer Science (FOCS)}, pages 177--193. IEEE, 2023.

\bibitem[AS23b]{abbe2023reed}
Emmanuel Abbe and Colin Sandon.
\newblock Reed--{M}uller codes have vanishing bit-error probability below capacity: a simple tighter proof via camellia boosting.
\newblock arXiv:2312.04329, 2023.

\bibitem[ASSY23]{ASSY23}
Emmanuel Abbe, Ori Sberlo, Amir Shpilka, and Min Ye.
\newblock Reed--{M}uller codes.
\newblock {\em Foundations and Trends in Communications and Information Theory}, 20(1--2):1--156, 01 2023.

\bibitem[BV04]{boyd2004convex}
Stephen Boyd and Lieven Vandenberghe.
\newblock {\em Convex Optimization}.
\newblock Cambridge University Press, 2004.

\bibitem[CR21]{couvreur2021algebraic}
Alain Couvreur and Hugues Randriambololona.
\newblock Algebraic geometry codes and some applications.
\newblock In {\em Concise Encyclopedia of Coding Theory}, pages 307--362. Chapman and Hall/CRC, 2021.

\bibitem[DMZ23]{dong2023number}
Dingding Dong, Nitya Mani, and Yufei Zhao.
\newblock On the number of error correcting codes.
\newblock {\em Combinatorics, Probability and Computing}, 32(5):819--832, 2023.

\bibitem[Eli91]{Eli91}
Peter Elias.
\newblock Error-correcting codes for list decoding.
\newblock {\em IEEE Transactions on Information Theory}, 37(1):5--12, 1991.

\bibitem[GHK10]{GHK10}
Venkatesan Guruswami, Johan Håstad, and Swastik Kopparty.
\newblock On the list-decodability of random linear codes.
\newblock In {\em Symposium on Theory of Computing (STOC)}, pages 409--416, 2010.

\bibitem[GRS25]{guruswami2012essential}
Venkatesan Guruswami, Atri Rudra, and Madhu Sudan.
\newblock Essential coding theory.
\newblock {\em Draft available at http://www. cse.buffalo.edu/faculty/atri/courses/coding-theory/book/}, 2025.

\bibitem[HP94]{HP94}
Hanan Herzberg and Gregory Poltyrev.
\newblock Techniques of bounding the probability of decoding error for block coded modulation structures.
\newblock {\em IEEE Transactions on Information Theory}, 40(3):903--911, 1994.

\bibitem[HSS21]{hkazla2021codes}
Jan H{\k{a}}z{\l}a, Alex Samorodnitsky, and Ori Sberlo.
\newblock On codes decoding a constant fraction of errors on the {BSC}.
\newblock In {\em Symposium on Theory of Computing (STOC)}, pages 1479--1488, 2021.

\bibitem[Iha81]{ihara1981some}
Yasutaka Ihara.
\newblock Some remarks on the number of rational points of algebraic curves over finite fields.
\newblock {\em Journal of the Faculty of Science, University of Tokyo}, 28(3):721--724, 1981.

\bibitem[Joh62]{Joh62}
Selmer Johnson.
\newblock A new upper bound for error-correcting codes.
\newblock {\em IRE Transactions on Information Theory}, 8(3):203--207, 1962.

\bibitem[KKM{\etalchar{+}}16]{kudekar2016reed}
Shrinivas Kudekar, Santhosh Kumar, Marco Mondelli, Henry~D Pfister, Eren {\c{S}}a{\c{s}}o{\u{g}}lu, and R{\"u}diger Urbanke.
\newblock Reed-{M}uller codes achieve capacity on erasure channels.
\newblock In {\em Symposium on Theory of Computing (STOC)}, pages 658--669, 2016.

\bibitem[LM00]{laurent2000adaptive}
Beatrice Laurent and Pascal Massart.
\newblock Adaptive estimation of a quadratic functional by model selection.
\newblock {\em Annals of Statistics}, pages 1302--1338, 2000.

\bibitem[PB23]{pathegama2023smoothing}
Madhura Pathegama and Alexander Barg.
\newblock Smoothing of binary codes, uniform distributions, and applications.
\newblock {\em Entropy}, 25(11):1515, 2023.

\bibitem[Plo60]{Plo60}
Morris Plotkin.
\newblock Binary codes with specified minimum distance.
\newblock {\em IRE Transactions on Information Theory}, 6(4):445--450, 1960.

\bibitem[Pol94]{poltyrev2002bounds}
Gregory Poltyrev.
\newblock Bounds on the decoding error probability of binary linear codes via their spectra.
\newblock {\em IEEE Transactions on Information Theory}, 40(4):1284--1292, 1994.

\bibitem[PSW25]{pernice2025list}
Francisco Pernice, Oscar Sprumont, and Mary Wootters.
\newblock List-decoding capacity implies capacity on the q-ary symmetric channel.
\newblock In {\em Symposium on Theory of Computing (STOC)}, pages 855--866, 2025.

\bibitem[PSZ25]{PSZ25}
Henry~D Pfister, Oscar Sprumont, and Gilles Zémor.
\newblock From bit to block: {D}ecoding on erasure channels.
\newblock In {\em International Symposium on Information Theory (ISIT)}, pages 1--6, 2025.

\bibitem[RP23]{reeves2023reed}
Galen Reeves and Henry~D Pfister.
\newblock Reed--{M}uller codes on {BMS} channels achieve vanishing bit-error probability for all rates below capacity.
\newblock {\em IEEE Transactions on Information Theory}, 70(2):920--949, 2023.

\bibitem[RU08]{RU08}
Tom Richardson and Rüdiger Urbanke.
\newblock {\em Modern Coding Theory}.
\newblock Cambridge University Press, 2008.

\bibitem[RU10]{rudra2010two}
Atri Rudra and Steve Uurtamo.
\newblock Two theorems on list decoding.
\newblock In {\em Approximation, Randomization, and Combinatorial Optimization (APPROX/RANDOM)}, pages 696--709, 2010.

\bibitem[Sam19]{samorodnitsky2019upper}
Alex Samorodnitsky.
\newblock An upper bound on $\ell_q $ norms of noisy functions.
\newblock {\em IEEE Transactions on Information Theory}, 66(2):742--748, 2019.

\bibitem[Sam20]{samorodnitsky2020improved}
Alex Samorodnitsky.
\newblock An improved bound on $\ell_q $ norms of noisy functions.
\newblock arXiv:2010.02721, 2020.

\bibitem[SS06]{sason2006performance}
Igal Sason and Shlomo Shamai.
\newblock Performance analysis of linear codes under maximum-likelihood decoding: {A} tutorial.
\newblock {\em Foundations and Trends in Communications and Information Theory}, 3(1--2):1--222, 2006.

\bibitem[TVZ82]{tsfasman1982modular}
MA~Tsfasman, SG~Vl{\u{a}}duţ, and Th~Zink.
\newblock Modular curves, {S}himura curves, and {G}oppa codes, better than {V}arshamov--{G}ilbert bound.
\newblock {\em Mathematische Nachrichten}, 109(1):21--28, 1982.

\bibitem[TZ00]{tillich2000discrete}
Jean-Pierre Tillich and Gilles Z{\'e}mor.
\newblock Discrete isoperimetric inequalities and the probability of a decoding error.
\newblock {\em Combinatorics, Probability and Computing}, 9(5):465--479, 2000.

\end{thebibliography}
